\documentclass{article}
\usepackage[utf8]{inputenc}
\usepackage[margin=1in]{geometry}
\usepackage{amsfonts,amsmath,amssymb,amsthm,mathtools,stmaryrd,tikz}
\usetikzlibrary{positioning}
\usepackage{dsfont}

\usepackage{hyperref}
\usepackage[dvipsnames]{xcolor}
\hypersetup{
 colorlinks=true,
 urlcolor=MidnightBlue,
 linkcolor=MidnightBlue,
 citecolor=MidnightBlue,
 unicode
}	

\usepackage{tocloft}
\usepackage{dirtytalk}
\usepackage[nottoc]{tocbibind} 
\usepackage{xparse}
\usepackage{enumitem}
\setlist{listparindent=\parindent,parsep=0pt,itemsep=1em}
\setlist[itemize]{label=$-$,noitemsep}
\setlist[enumerate]{itemsep=1mm}
\setlist[description]{leftmargin=\parindent}

\usepackage[lined,boxed,commentsnumbered, ruled]{algorithm2e}

\usepackage[capitalise,nameinlink,noabbrev]{cleveref}
\usepackage{algpseudocode, graphicx, times, thmtools, thm-restate}
\usepackage[noadjust]{cite}
\usepackage[full]{complexity}
\usepackage{tocloft}
\definecolor{betterYellow}{RGB}{255,255,0}

\declaretheorem[name=Theorem, numberwithin=section]{theorem}
\declaretheorem[name=Corollary, sibling=theorem]{corollary}
\declaretheorem[name=Proposition, sibling=theorem]{proposition}
\declaretheorem[name=Lemma, sibling=theorem]{lemma}

\declaretheorem[name=Observation, sibling=theorem]{observation}

\declaretheorem[name=Notation, sibling=theorem]{notation}
\theoremstyle{definition}

\declaretheorem[name=Definition, sibling=theorem]{definition}

\theoremstyle{remark}

\newcommand{\vars}{\ensuremath{\mathsf{Vars}}}

\newcommand{\depth}{\ensuremath{\mathsf{depth}}}

\newcommand{\size}{\ensuremath{\mathsf{size}}}

\renewcommand{\include}{\input}

\newcommand{\mycomment}[1]{}
\newcommand{\MIN}{\Min}

\newcommand{\newprob}[2]{\newcommand{#1}{{\text{\upshape\scshape #2}}\xspace}}
\newprob{\Pigeon}{Pigeon}
\newprob{\WeakPigeon}{Weak-Pigeon}
\newprob{\Nash}{Nash}
\newprob{\PE}{PE}
\newprob{\GPLS}{GPLS}
\newprob{\EOL}{End-Of-Line}
\newprob{\Factoring}{Factoring}
\newprob{\Blichfeldt}{Blichfeldt}
\newprob{\SVP}{SVP}
\newprob{\EKR}{Erd\"{o}s-Ko-Rado}
\newprob{\Ramsey}{Ramsey}
\newprob{\LongChoice}{Long-Choice}
\newprob{\Sunflower}{Sunflower}
\newprob{\FF}{FF}
\newprob{\Search}{Search}
\newprob{\Refl}{Ref}
\newprob{\Sat}{Sat}
\newprob{\Proof}{Proof}
\newprob{\Iter}{Iter}
\newprob{\rWPHP}{rWPHP}
\newprob{\TreeIter}{TreeIter}
\newprob{\LOP}{LOP}
\newprob{\sRA}{StrongAvoid}
\newprob{\wRA}{Avoid}
\newprob{\sRAlong}{strongRangeAvoidance}
\newprob{\wRAlong}{RangeAvoidance}
\newprob{\SoD}{SoD}
\newprob{\SoDlong}{SourceofDag}
\newprob{\USoD}{USoD}
\newprob{\USoDlong}{UnmeteredSourceofDag}
\newprob{\Shattering}{Shattering}
\newprob{\CPLS}{CPLS}
\newprob{\SoL}{SoL}
\newprob{\SoPL}{SoPL}
\newprob{\rPHP}{rPHP}
\newprob{\GIter}{GIter}
\newprob{\GSoPL}{GSoPL}
\newprob{\Empty}{Empty}
\newprob{\Min}{LeastNumber}

\newclass{\SOD}{SOD}
\newclass{\PWPP}{PWPP}
\newclass{\PLC}{PLC}
\newclass{\UPLC}{UPLC}
\newclass{\EOPL}{EOPL}
\newclass{\CLS}{CLS}
\newclass{\TFPH}{TFPH}
\newclass{\TFZPP}{TFZPP}
\newclass{\TF}{TF}
\newclass{\F}{F}
\newclass{\APEPP}{APEPP}
\newclass{\PEPP}{PEPP}
\newclass{\SOPL}{SOPL}
\newclass{\SOL}{SOL}
\newclass{\promise}{pr}
\newcommand{\pr}[1]{\promise#1}

\newclass{\Res}{Res}
\newclass{\Rev}{RevRes}
\newclass{\Circ}{uCircRes}
\newclass{\PK}{PK}
\newclass{\LK}{LK}
\newclass{\TreeRes}{TreeRes}
\newclass{\SA}{uSA}
\newclass{\wSA}{wSA}
\newclass{\resc}{Res(\C)}
\newclass{\sAvoid}{StrongAvoid}
\newclass{\LOPclass}{LOP}
\newclass{\MINclass}{LeastNumber}

\newclass{\CNF}{CNF}

\NewDocumentCommand\bra{mo}{
  \IfNoValueTF{#2}
  {[\![ #1]\!]}
 {[\![ #1]\!]_{#2}}}

\NewDocumentCommand\pe{o}{
  \IfNoValueTF{#1}
  {\tilde{\mathbb{E}}}
 {\tilde{\mathbb{E}}\left[#1\right]}}

\newcommand{\bn}{\{0,1\}^n} 
\DeclareMathOperator{\Ord}{Ord}

\newcommand{\rest}{\!\!\upharpoonright\!}
\newcommand{\setm}{\!\setminus\!}

\title{Separations above TFNP from Sherali-Adams Lower Bounds}

\begin{document}

\newgeometry{margin=1.3in,top=1.7in,bottom=1in}

\begin{center}
{\huge Separations above TFNP from Sherali-Adams Lower Bounds}
\\[9mm] \large

\setlength\tabcolsep{1em}
\begin{tabular}{cccc}
Noah Fleming&
Anna G\'al&
Deniz Imrek&
Christophe Marciot\\[-1mm]
\small\slshape Lund \& Columbia &
\small\slshape UT Austin &
\small\slshape UT Austin &
\small\slshape Lund University 
\end{tabular}

\vspace{9mm}

\large

\vspace{5mm}

\normalsize
\end{center}

\begin{abstract}
    Unlike in $\TFNP$, for which there is an abundance of problems capturing natural existence principles which are incomparable (in the black-box setting), Kleinberg et al.~\cite{KleinbergKMP21} 
    observed that many of the natural problems considered so far in the second level of the total function polynomial hierarchy $(\TF\Sigma_2)$ reduce to the Strong Avoid problem.
    In this work, we prove that the Linear Ordering Principle does not reduce to Strong Avoid in the black-box setting, exhibiting the first $\TF\Sigma_2$ problem that lies outside of the class of problems reducible to Strong Avoid.

    The proof of our separation exploits a connection between total search problems in the polynomial hierarchy and proof complexity, recently developed by Fleming, Imrek, and Marciot \cite{FlemingIM25}. In particular, this implies that to show our separation, it suffices to show that there is no small proof of the Linear Ordering Principle in a $\Sigma_2$-variant of the Sherali-Adams proof system. To do so, we extend the classical pseudo-expectation method to the $\Sigma_2$ setting, showing that the existence of a $\Sigma_2$ pseudo-expectation precludes a $\Sigma_2$ Sherali-Adams proof. The main technical challenge is in proving the existence of such a pseudo-expectation, we manage to do so by solving a combinatorial covering problem about permutations. We also show that the extended pseudo-expectation bound implies that the Linear Ordering Principle cannot be reduced to any problem admitting a low-degree Sherali-Adams refutation.

\end{abstract}
\vspace{8mm}

\setlength{\cftbeforesecskip}{0pt}
\renewcommand\cftsecfont{\mdseries}
\renewcommand{\cftsecpagefont}{\normalfont}
\renewcommand{\cftsecleader}{\cftdotfill{\cftdotsep}}
\setcounter{tocdepth}{1}
\tableofcontents

\thispagestyle{empty}
\setcounter{page}{0}

\newpage
\restoregeometry

\section{Introduction}

In recent years total search problems in the second level of the polynomial hierarchy $(\TF\Sigma_2)$ have received considerable attention. A substantial reason for this is that this class contains a variety of important explicit construction problems, such as finding truth tables of functions with high circuit complexity, pseudo-random generators, rigid matrices, time-bounded Kolmogorov random strings, and extractors. The totality of these problems is witnessed by union-bound type arguments, which can be formalized as reductions to the $\TF\Sigma_2$ problem $\wRA$ \cite{Korten21, KleinbergKMP21}. 
\paragraph{Avoid.} Given $C: [2^n] \rightarrow [2^{n+1}]$ output $y \in [2^{n+1}]$ such that for every $x \in [2^n]$, $C(x) \neq y$.\\
\mbox{}\\
\noindent By developing new algorithms for this problem, as well as the harder Linear Ordering Principle, researchers have obtained state-of-the-art circuit lower bounds \cite{Li24, ChenHRMSO24, KortenP24} and data structure lower bounds \cite{GuruswamiLW25,KortenPI25}. 

\paragraph{Linear Ordering Principle (\LOP).} Given $\prec: [2^n] \times [2^n] \rightarrow \{0,1\}$, output either\vspace{-0.4em}
\begin{itemize}
  \item $x$ such that for every $y \neq x$, $x \prec y$, or \hfill \emph{(Minimal element)}
  \item $x \neq y \neq z$ such that either (i) $x \prec x$, or (ii) $x \not \prec y$ and $y \not \prec x$, or (iii) $x \prec y, y \prec z$ but $x \not \prec z$. \\ 
  \mbox{}\hfill \emph{(Linear ordering violation)}
\end{itemize}

Kleinberg et al.~\cite{KleinbergKMP21} initiated the study of $\TF\Sigma_2$ 
as a class, introducing many of the aforementioned explicit construction problems along with a number of other natural search problems, including 
$\mathsf{Ramsey}$-$\mathsf{Erdos}$$\hspace{-10pt}$ \textsf{\H{}}$\hspace{3.5pt}$ $\mathsf{Completion}$ which captures the proof that the $n$-th Ramsey number is at least $2^{n/2}$. In doing so they observed that, 
while $\TFNP$ has a variety of incomparable subclasses capturing various existence principles, every $\TF\Sigma_2$ problem that they considered admitted a reduction to the following strong variant of the $\wRA$ problem.

\paragraph{Strong Avoid.} Given $C: [2^n] \rightarrow [2^n+1]$ output $y \in [2^n+1]$ such that for every $x \in [2^n]$, $C(x) \neq y$.\\
\mbox{}\\
\cite{KleinbergKMP21} raised the question of whether there are any natural problems in $\TF\Sigma_2$ which do not reduce to the problem $\sRA$.
The main contribution of our work is to give the first $\TF\Sigma_2$ problem which is not contained within $\sAvoid^{dt}$, the class of problems which admit efficient black-box reductions to $\sRA$. This answers the question in the black-box setting. Note that a separation in the Turing Machine setting would separate $\P$ from $\Sigma_2^P$. 

\begin{theorem}
\label{thm:LOPnotsRA}
    $\LOP \not \in \sAvoid^{dt}$.
\end{theorem}
This complements the work of Korten and Pitassi \cite{KortenP24}, who show a separation in the other direction. 
Using our technique, we are also able to show that $\LOP$ does not reduce to the $\TF\Sigma_2$ problem $\Min$, introduced by Thapen \cite{Thapen2024}.

\paragraph{Least Number.} Given $f:\bn\rightarrow \{0,1\}$, output either $\bot$ if $f(x)=0$ for all $x$, or $x$ such that $f(x)=1$ and for all $y<x, f(y)=0$. 
\mbox{}\\

\begin{theorem}
\label{thm:LOPvsLN}
    $\LOP \not \in \MINclass^{dt}$. 
\end{theorem}

Our separations, along with those known already in the literature, are represented pictorially in \autoref{fig:results}.

    \begin{figure}[h]
       \begin{center}
           \begin{tikzpicture}

                \draw[->, thick, color=black!60, very thick] (-0.2,-0.3) -- (-1.575,1.5);
                \draw[->, thick, color=black!60, very thick] (0.2,-0.3) -- (2.4,3.3);
                \draw[->, thick, color=black!60, very thick] (3,0.3) -- (3,3.3);
                \draw[->, thick, color=black!60, very thick] (-2.025,2.1) -- (-2.9,3.3);

                \draw[->, dotted, color=red!50!purple, very thick] (3.2,3.3) -- (3.2,0.3);
                \draw[->, dotted, color=red!50!purple, very thick] (1.7,3.4) to[out=190,in=-10] (-2.5,3.4);
                \draw[->, dotted, color=red!50!purple,  very thick] (0.5,-0.6) -- (2.4,0);

                \draw[->, dashed, thick, color=red!50!purple, very thick] (-2.7,3.3) -- (-1.825, 2.1);
               \draw[->, dashed, thick, color=red!50!purple, very thick] (-2.5,3.8) to[out=10,in=170](1.7,3.8);

               \fill[draw=black!40!white, very thick, rounded corners, fill=Cerulean!10!white] (-0.5, -0.9) rectangle (0.5, -0.3) node[pos=.5] {~$\FNP~$};
               \fill[draw=black!40!white, very thick, rounded corners, fill=Cerulean!10!white] (2.4, -0.3) rectangle (3.6, 0.3) node[pos=.5] {~$\wRA~$};
               \fill[draw=black!40!white, very thick, rounded corners, fill=Cerulean!10!white] (1.7, 3.3) rectangle (4.3, 3.9) node[pos=.5] {~$\sRA~$};
               \fill[draw=black!40!white, very thick, rounded corners, fill=Cerulean!10!white] (-3.5, 3.3) rectangle (-2.5, 3.9) node[pos=.5] {~$\LOP~$};
               \fill[draw=black!40!white, very thick, rounded corners, fill=Cerulean!10!white] (-2.8, 1.5) rectangle (-0.2, 2.1) node[pos=.5] {~$\Min~$};
           \end{tikzpicture}
       \end{center}
       \caption{Relationships of some $\TF\Sigma_2^{dt}$ classes. A black arrow from a class $A$ to a class $B$ means that $A\subseteq B$ (\cite{KleinbergKMP21}). A dashed or dotted arrow from a class $A$ to a class $B$ means that $A\not\subseteq B$. The dashed separations are proved in this paper. \cite{KortenP24} proves that $\sRA$ does not reduce to $\wRA$ and $\LOP$. \cite{FlemingGRJLSY25} proves that $\wRA^{dt}$ does not contain all of $\TFNP^{dt}$. This implies that none of the other problems in the diagram reduces to $\wRA$.}
       \label{fig:results}
    \end{figure}

\subsection{Technical Highlights}

To prove \autoref{thm:LOPnotsRA} we make use of a connection between $\TF\Sigma_2^{dt}$ (where the superscript $dt$ denotes black-box reductions) and proof complexity, developed in \cite{FlemingIM25}. In doing so we provide the first proof of a  
separation between black-box $\TF\Sigma_2$ classes 
that uses proof complexity, confirming that this connection can indeed be useful. In particular, \cite{FlemingIM25} shows that black-box reductions to $\sRA$ are equivalent to efficient proofs in a $\Sigma_2$-variant of the Sherali-Adams proof system (see \hyperref[def:uSAproofsyst]{Definition~\ref*{def:uSAproofsyst}}) which equips it with a $\Sigma_2$-weakening step, extending the weakening rule for resolution to depth-$2$ formulas, as follows. 
\begin{definition}
\phantomsection\label{def:weakening}
    Let $D$ be a DNF formula, a \emph{$\Sigma_2$-weakening} of $D$ is a collection of DNFs $\{D_i\}_{i \in [m]}$ such that $D \implies D_i$ for every $i \in [m]$. A $\Sigma_2$-weakening of a collection of DNFs is a union of $\Sigma_2$-weakening of those DNFs.
\end{definition}

This connection with proof complexity reduces proving \autoref{thm:LOPnotsRA} to showing that the unsatisfiable formula which encodes the totality of $\LOP$
does not have efficient $\Sigma_2$-Sherali-Adams proofs (\hyperref[def:uSAproofsyst]{Definition~\ref*{def:uSAproofsyst}}).

\subsubsection*{Pseudo-Expectations}
To prove our Sherali-Adams lower bound, we extend the well-developed pseudo-expectation technique for Sherali-Adams proofs to the $\Sigma_2$ setting. 
A pseudo-expectation is an object which appears to be a probability distribution over satisfying assignments to our (unsatisfiable) formula when examining low-degree marginals. The existence of a pseudo-expectation precludes a low-degree Sherali-Adams proof. A $\Sigma_2$-pseudo-expectation for $\LOP$ is a collection of pseudo-expectations, one for each set of weakenings of the constraints of $\LOP$. See \hyperref[def:pseudoexp]{Definition~\ref*{def:pseudoexp}} for details. The main technical challenge is to construct such a $\Sigma_2$-pseudo-expectation. While Dantchev et al.~\cite{DantchevMR09} showed that $\LOP$ itself has a high-degree pseudo-expectation, it is not clear how to extend their argument to all $\Sigma_2$-weakenings. 

\subsubsection*{A Covering Problem of Permutations}
We observe that constructing a $\Sigma_2$-pseudo-expectation requires us to solve the following covering problem.
Let $\Ord$ be the set of all total orders on the set of numbers $\{1,\ldots, n\}$, and let $\Ord^{*1}\subseteq\Ord$ be the set of total orders for which at least one number comes before the element 1. For a fixed $d\leq n$, a set $S \subseteq [n]$ of size $d$ and an order $\sigma$ on $S$, 
let $\cal{C}_{S, \sigma}= \{\pi\in \Ord\mid \pi\text{ induces the order $\sigma$ on $S$}\}$ be the set of all total orders on $[n]$ that are consistent with $\sigma$ on $S$. Is there a set of strictly less than $d!$ pairs $(S_i, \sigma_i)$, with each $S_i$ of size $d$ such that $\Ord^{*1}$ is contained in
the union $\bigcup_i \mathcal{C}_{S_i, \sigma_i}$?

Note that one can cover the set of all total orders with exactly $d!$ such collections $\cal{C}_{S, \sigma}$.
The question is thus "can we do better if we do not need to cover orders starting with 1?". One can extract from the lower bound of Dantchev et al.~\cite{DantchevMR09} that for $d=2$, such coverings are impossible. Our proof of \autoref{thm:LOPnotsRA} implies that the answer to this question is no for any $d \leq n/100$. 
In our argument, we show that weakenings of some axioms of $\LOP$ can be modified to a normalized form without increasing their degree by much. 
Proving that our pseudo-expectation function is non-negative on these normalized weakenings is equivalent to a lower bound in the above covering problem.

Our linear lower bound on $d$ is tight up to constant factors, as one can construct a covering with less than $d!$ collections for $d=n/2$.

\subsubsection*{Further Separations from Sherali-Adams Upper Bounds} 

Leveraging our $\Sigma_2$-pseudo-expectation, we show that if a $\TF\Sigma_2$ problem admits a Sherali-Adams upper bound then this problem is separated from the Linear Ordering Principle. Together with a simple Sherali-Adams upper bound for the $\Min$ we obtain \autoref{thm:LOPvsLN}. To prove this, we show that reductions can be split in two parts: first, a potentially hard to verify $\Sigma_2$-weakening step, followed by a \emph{counter-example reduction}, a type of reduction introduced in \cite{KolodziejczykT22,Thapen2024} that actually happens in $\TFNP^{dt}$. 
This reveals how the two types of black-box reductions that have been studied in the literature for total functions in the polynomial hierarchy relate.

\subsection{Related Works}
Ghentiyala et al.~\cite{Ghentiyala25} showed that $\LOP$ does not reduce to $\wRA$, even in the non-blackbox setting, assuming that $\P^{\pr{\MA}}\not\subset\AM\cap\coAM$. 
The only other known separation in $\TF\Sigma_2^{dt}$ was proven by Korten and Pitassi~\cite{KortenP24} who showed that the problem $\sRA \not \in \LOPclass^{dt}$. Their separation was obtained by proving a lower bound for depth-$3$ circuits via a switching lemma. This contrasts with our separation which uses proof complexity techniques. 

Proof complexity techniques have been the main method for obtaining separations between classes in $\TFNP$ in the black-box setting. This has led to a complete understanding of the relationships between the major classes \cite{BeameCEIP98, Morioka01, Buresh-OppenheimM04, GoosKRS19, GoosHJMPRT22, FlemingGPR24}. Proof complexity has been useful in $\TFNP$ because of the fact that membership in (sufficiently uniform) $\TFNP^{dt}$ subclasses is equivalent to efficient provability in some associated proof system \cite{BussFI23}. This has been used to show that, for example, the class $\PPADS^{dt}$ is equivalent to the unary Sherali-Adams proof system \cite{GoosHJMPRT22}. This connection was recently extended to subclasses of the total search problems in the polynomial hierarchy in \cite{FlemingIM25}, which is the basis for our work.

Recently, Thapen in \cite{Thapen2024} exhibited a new type of reductions, namely \emph{counter-example reductions}. They, along with the notion of \emph{herbrandization}, are used to find a $\TFNP^{dt}$ translation of what happens at the second level $\TF\Sigma_2^{dt}$ and they can also be used to define $\TFNP^{dt}$ subclasses from $\TF\Sigma_2^{dt}$ problems.
In \hyperref[sec:furtherSep]{Section~\ref*{sec:furtherSep}}, we add a bridge between this notion of counter-example reduction and the notion of $\Sigma_2$-weakening defined in \cite{FlemingIM25} and show that they are sufficient ingredients to capture reductions.

\section{Preliminaries on Black-Box $\TF\Sigma_2$}\label{sec:Prelim}
A \emph{query search problem} is a sequence of relations $R_n \subseteq \{0,1\}^n \times {\cal O}_n$, one for each $n \in \mathbb{N}$. It is \emph{total} if for every $x \in \{0,1\}^n$ there is an $o \in {\cal O}_n$ such that $(x,o) \in R_n$. We think of $x \in \{0,1\}^n$ as a bit string which can be accessed by querying individual bits, and we will measure the complexity of solving $R_n$ as the number of bits that must be queried. Hence, an efficient algorithm for $R_n$ will be one which finds a suitable $o$ while making at most $\polylog (n)$-many queries to the input. We will not charge the algorithm for other computational steps, and therefore an efficient algorithm corresponds to a shallow decision tree. Total query search problems which can be computed by decision trees of depth $\polylog(n)$ belong to the class $\FP^{dt}$, where $dt$ indicates that it is a black-box class. While search problems are formally defined as sequences $R=(R_n)$, we will often want to speak about individual elements in the sequence. For readability, we will abuse notation and refer to elements $R_n$ in the sequence as total search problems. In addition, we will often drop the subscript $n$, and rely on context to differentiate.

In this paper we will be considering total query search problems in the second level of the polynomial hierarchy $\TF\Sigma_2^{dt}$.

\begin{definition} \label{def:void} A total query search problem $R=(R_n)$, where $R_n \subseteq \{0,1\}^n \times {\cal O}_n$, belongs  to $\TF\Sigma_2^{dt}$ if for every $n$, $o \in {\cal O}_n$ and $x\in\{0,1\}^n$
\[ R_n(x,o) \iff \forall z \in \{0,1\}^{\ell}, V_{o,z}(x)=1,\]
where $\ell = \polylog(n)$ and for every $z \in \{0,1\}^{\ell}$, $V_{o,z}$ is a $\polylog(n)$-depth decision tree.

\end{definition}

We can compare the complexity of total search problems by taking reductions between them. The following defines \emph{decision tree reductions} (known as \emph{formulations}) between total query search problems, the query analogue of polynomial-time reductions.

\begin{definition}
\label{def:many_one_red}
    For total query search problems $R \subseteq \{0,1\}^n \times {\cal O}_{n}, S \subseteq \{0,1\}^m \times {\cal O}_{m}'$, there is an \emph{$S$-formulation of $R$} if, for every $i \in [m]$ and $o \in {\cal O'}_m$, there are functions $f_i :\{0,1\}^n \rightarrow \{0,1\}$ and $g_o : \{0,1\}^n \rightarrow {\cal O}_n$ such that 
    \[ S(f(x),o) \implies R(x,g_o(x)),\]
      where $f(x) = (f_1(x) \ldots f_m(x))$. The \emph{depth} of the formulation is 
      \[d~:=~\max\big (\{ \depth(f_i) : i \in [m]\} \cup \{\depth(g_o): o \in {\cal O}_m'\} \big),\]
      where $\depth( f )$ denotes the minimum depth of any decision tree which computes $f$. The \emph{size} of the formulation is $m$, the number of input bits to $S$. The \emph{complexity} of the formulation is $\log m+d$.
      The \emph{complexity of reducing $R$ to $S$} is the minimum complexity of all possible $S$-formulations of $R$.

      We extend this definition to sequences in the natural way. If $S=(S_m)$ is a sequence and $R_n$ is a single search problem, then the complexity of reducing $R_n$ to $S$ is the minimum over $m$ of the complexity of reducing $R_n$ to $S_m$. For two sequences of search problems $S=(S_m)$ and $R=(R_n)$, the complexity of reducing $R$ to $S$ is the complexity of reducing $R_n$ to $S$, as a function of $n$. A reduction from $R$ to $S$ is efficient if its complexity is $\polylog(n)$; we denote this by $R \leq_{dt} S$. 
\end{definition}

 Black-box $\TF\Sigma_2$ can be viewed as the study of the following family of total search problems. This will allow us to leverage a close connection between $\TF\Sigma_2^{dt}$ and proof complexity \cite{FlemingIM25} to prove our main result.

\begin{definition}
	Let $F = D_1 \wedge \ldots \wedge D_t$ be a formula in which each $D_i$ is a DNF. The \emph{false formula} search problem $\FF_{F} \subseteq \{0,1\}^n \times [t]$ is defined as 
\[ \FF_{F}(x,o) \iff D_o(x)=0.\]
\end{definition}
\begin{lemma}[\cite{FlemingIM25}]\label{lem:FFequiv}
	For any $R \in \TF\Sigma_2^{dt}$ there is an unsatisfiable formula $F_R = D_1 \wedge \ldots \wedge D_t$ where each $D_i$ is a DNF of $\polylog(n)$-width, such that $R =_{dt} \FF_{F_R}$. 
\end{lemma}
 This extends the well-known connection between $\TFNP^{dt}$ and the false clause search problem. 
The proof of this lemma proceeds by writing down the totality of $R$ as a $\Sigma_3$-formula and then taking its negation to obtain $F_R$.
Throughout this paper, we will abuse notation and use $R$ to refer to $F_R$, using context to differentiate. Hence, $\FF_R = \FF_{F_R}$.

\section{Proof Complexity and a Sufficient Condition for Separations}
To prove \autoref{thm:LOPnotsRA} we make use of a connection between $\TF\Sigma_2^{dt}$ and proof complexity. 
Fleming, Imrek, and Marciot \cite{FlemingIM25} showed that decision tree reductions between $\TF\Sigma_2^{dt}$ problems are equivalent to efficient proofs in certain proof systems equipped with the $\Sigma_2$-weakening rule from \hyperref[def:weakening]{Definition~\ref*{def:weakening}}. Using this connection they showed that proofs in the unary $\Sigma_2$-Sherali-Adams proof system are equivalent to reductions to $\sRA$, which we describe next.

For any boolean formula $F$, we will assume without loss of generality that all negations occur at the leaves and let $\vars^+(F)$ be the positive literals in $F$ and $\vars^-(F)$ be the negative literals. For any conjunct $t = \bigwedge_{x_i \in \vars^+(t)} x_i \wedge \bigwedge_{x_j \in \vars^-(t)}\neg x_j$, we associate the following polynomial $t(x)=\prod_{x_i \in \vars^+(t)} x_i \prod_{x_j \in \vars^-(t)}(1-x_j)$. 
We refer to the polynomials $t(x)$ also as conjuncts and also denote them by $t$.
We say that a \emph{conical junta} is a sum of conjuncts ${\cal J}:= \sum_i t_i$.

Let $D =\bigvee_{t \in D} t$ be any DNF. We can express $D$ as a polynomial
\begin{equation}\label{eq:poly}
 D(x) = \sum_{t \in D} t(x) -1.
\end{equation}
The degree of this polynomial is $\deg(D):= \max_{t \in D} \deg(t)$. 
We refer to $\deg(D)$ as the \emph{degree of the DNF} $D$.
 Observe that for any $x \in \{0,1\}^n$, the DNF  $D(x)$ is true iff $\sum_{t \in D}t(x) -1 \geq 0$. 
 We will abuse notation and denote by $D$ both the DNF and the associated polynomial, using context to differentiate.

Throughout we will work with \emph{multi-linear arithmetic}, associating $x_i^2=x_i$ for every variable $x$. This has the effect of restricting the underlying linear program to $\{0,1\}$-points. 
\begin{definition}
\phantomsection
\label{def:uSAproofsyst}
    Let $F = \{D_i\}_{i \in [m]}$ be an unsatisfiable collection of DNFs (a collection with no shared satisfying assignment).
     A $\Sigma_2$-\emph{unary-Sherali-Adams} (denoted $\SA$) proof $\Pi$ of $F$ is a $\Sigma_2$-weakening $F' = \{D_i'\}_{i \in [m']}$ of $F$ 
    together with a list of conical juntas ${\cal J}_i, {\cal J}$, such that
    \[ \sum_{i \in [m']} D_i' {\cal J}_i + {\cal J} = -1. \]
     The \emph{degree} $\deg(\Pi)$ is the maximum degree among $D_i, D_i' {\cal J}_i,$ and ${\cal J}$, and the \emph{size} $\size(\Pi)$ is the number of monomials 
     in $D_i,D_i' {\cal J}_i, {\cal J}$
     counted with multiplicity (i.e., each monomial is counted again each time it appears). 
    The \emph{complexity} of the proof is given by $\SA(\Pi)~:=~ \deg(\Pi)+\log \size(\Pi)$, and the complexity of proving $F$ is $\SA(F)~:=~\min_{\Pi} \SA(\Pi)$, where the minimum is taken over all $\SA$ proofs $\Pi$ of $F$.
\end{definition}
The only difference between the standard Sherali-Adams proof system and its unary variant is the measure of size; in particular, we cannot use large coefficients. When measuring degree, these systems are identical.

To prove \autoref{thm:LOPnotsRA} it suffices to show that there is no efficient $\Sigma_2$-$\SA$ proof of $\LOP$. 
\begin{theorem}[\cite{FlemingIM25}]\phantomsection\label{thm:wRACharact}
    For any $R \in \TF\Sigma_2^{dt}$ there exists a $\polylog(n)$-complexity $\Sigma_2$-$\SA$ proof of $R$ iff $R \in \sAvoid^{dt}$. 
\end{theorem}

To prove our lower bound for $\LOP$ we generalize the method of pseudo-expectations, which has been developed for Sherali-Adams (see \cite{FlemingKP19} for a survey), to $\Sigma_2$-Sherali-Adams degree lower bounds.

\begin{definition}\phantomsection\label{def:pseudoexp}
    For any collection of DNFs $F=\{D_i\}_{i \in [m]}$ a \emph{degree-$d$ pseudo-expectation} for $F$ is a linear function $\pe : \mathbb{R}[x] \rightarrow \mathbb{R}$ satisfying
    \begin{enumerate}
        \item $\pe[1] = 1$,
        \item $\pe [{\cal J}] \geq 0$ for every conical junta ${\cal J}$ of degree at most $d$,
        \item $\pe [D_i{\cal J}] \geq 0$ for every $D_i \in F$ and conical junta $J$ such that $\deg(D_i {\cal J}) \leq d$.
    \end{enumerate}
        A degree-$d$ \emph{$\Sigma_2$-pseudo-expectation} for $F$ is a family of degree-$d$ pseudo-expectations $\{\pe_G\}$, one for every $\Sigma_2$-weakening $G$ of $F$. 
\end{definition}

Note that a $\Sigma_2$-pseudo expectation for $F$ is equivalent to a standard pseudo-expectation for the collection of DNFs $F'$ which includes every $\Sigma_2$-weakening of the DNFs in $F$.

\begin{lemma}\phantomsection\label{lem:PESigma2}
    There is a degree-$d$ $\Sigma_2$-pseudo-expectation for $F$ iff there is no degree-$d$ $\Sigma_2$-$\SA$ proof of $F$. Moreover, if $F$ does not have a degree-$d$ $\Sigma_2$-$\SA$, then there is a degree-$d$ pseudo-expectation that works for all weakenings of $F$.
\end{lemma}

\begin{proof}
    The proof follows from the fact that for every unsatisfiable collection of DNFs (and indeed any collection of polynomials with no $\{0,1\}$-solutions), a degree-$d$ pseudo-expectation exists iff a Sherali-Adams proof does not (see e.g.~\cite{FlemingKP19}).
    Applying this to each weakening completes the proof. However, since we will use the forward direction of this theorem we include a proof.
    
    Let $\{\pe_G\}$ be a degree-$d$ $\Sigma_2$-pseudo-expectation and suppose for contradiction that there is also a degree-$d$ $\Sigma_2$-Sherali-Adams proof $(G, \Pi)$ of $F$, where $G= \{D_i'\}_{i \in [m']}$ is a $\Sigma_2$-weakening of $F$ and $\Pi$ is a list of conical juntas $\{{\cal J}_i\}, {\cal J}$ constituting a Sherali-Adams proof of $G$. Let $\pe_G$ be the pseudo-expectation which corresponds to $G$, then 
    \[ -1 = \pe_G[-1] = \pe \Big[\sum_{i \in [m']} D_i' {\cal J}_i + {\cal J} \Big] = \sum_{i \in [m']} \pe[D_i'{\cal J}_i] + \pe[{\cal J}] \geq 0.\] 

    Consider $\mathcal{G}=\bigcup_{G}G$, where the union is taken over every $\Sigma_2$-weakenings of $F$. Then $\pe_\mathcal{G}$ must be a $d$-pseudo-expectation for all weakenings of $F$, as they are all subsets of $\mathcal{G}$. 
\end{proof}

\section{Separating LOP from Strong Avoid} \label{sec:sepLOPfromSA}
In this section we prove the following.
\begin{theorem}
\phantomsection
\label{thm:PE}
   There exists $\Sigma_2$-pseudo-expectation for $\LOP$ of degree $(\frac{n}{400} - 1)$.
\end{theorem}

From this, the separation $\LOP \not \in \sAvoid^{dt}$ follows immediately. Note that \hyperref[thm:PE]{Theorem~\ref*{thm:PE}} is much stronger than we need in order to rule out a reduction to $\sRA$, as works irrespective of the size of the $\SA$ proofs, or size of the reduction.

\begin{proof}[Proof of \autoref{thm:LOPnotsRA}]
    By \hyperref[thm:wRACharact]{Theorem~\ref*{thm:wRACharact}} $\LOP \in \sAvoid^{dt}$ iff there exists a $\polylog(n)$-complexity $\Sigma_2$-$\SA$ proof of the propositional encoding of $\LOP$ given by \autoref{lem:FFequiv}. In particular, such a proof must have $\polylog(n)$ degree. However, by \hyperref[lem:PESigma2]{Lemma~\ref*{lem:PESigma2}}, the existence of a degree $\Omega(n)$ $\Sigma_2$-pseudo-expectation precludes the existence of such a $\Sigma_2$-$\SA$ proof.
\end{proof}

In the remainder of this section we prove \hyperref[thm:PE]{Theorem~\ref*{thm:PE}}.
Applying \autoref{lem:FFequiv} to $\LOP$, we obtain an unsatisfiable formula which consists of the following constraints (axioms) over variables $x_{i,j}$,
where $x_{i,j}=1$ means "$i\prec j$":
\begin{itemize}
  \item $M_i$: $\bigvee_{j\in[n]\setminus i} x_{j,i}$, for $i\in n$, \hfill \emph{(Non-minimality)}
  \item $R_{i}$: $\neg x_{i,i}$, for $i\in [n]$, \hfill\emph{(Irreflexivity)}
  \item $A_{i,j}$: $\neg x_{i,j}\vee \neg x_{j,i}$, for distinct $i,j\in[n]$, \hfill\emph{(Asymmetry)}
  \item $T_{i,j,k}$: $x_{i,k}\vee \neg x_{i,j}\vee \neg x_{j,k}$, for distinct $i,j,k\in[n]$, \hfill\emph{(Transitivity)}
  \item $O_{i,j}$: $x_{i,j}\vee \neg x_{i,j}$, for distinct $i,j\in[n]$. \hfill\emph{(Totality)}
\end{itemize}
Note that even though each of these axioms could be expressed as a clause, we consider them as DNFs of degree 1.
We will use the same pseudo-expectation for each weakening of $\LOP$, which will be the standard choice \cite{DantchevMR09}: a uniform distribution over total orders. Let $\Ord=\{x\in [n]^n\mid \forall i\neq j, x_i\neq x_j\}$ be the set of total orders on $[n]$. We associate total orders on $[n]$ with both 
  strings in $[n]^n$
  (simply listing the elements in the given order) and bijective maps $[n]\rightarrow [n]$. For $z\in \Ord$, and $T\subseteq [n]$, we define $z\rest T$ as the permutation on $T$, $\pi:|T|\rightarrow T$, that orders those elements in the same order as $z$ does. We denote $z\rest([n]\setm T)$ by $z\setm T$ for convenience. For $x \in \Ord$ and a monomial $t$ we denote by $t(x)$ the monomial $t$ evaluated
  over the variables $x_{i,j}$,
where $x_{i,j}=1$ means "$i\prec j$" (as above) in the order $x$.

  \begin{definition}
  We define our pseudo-expectation $\pe$ on monomials $t$ 
  as $$\pe[t]=\frac{|\{x\in \Ord\mid t(x)=1\}|}{|\Ord|}.$$
The definition is extended to arbitrary polynomials by linearity. 
\end{definition}

By our encoding of DNFs as polynomials (\autoref{eq:poly}), for a DNF $D = \bigvee_i t_i$, we use the
notation $\pe[D]$ to mean $\pe[\sum_i t_i-1]$. Note that $\pe[t+ t']$ corresponds to the number of orders either $t$ or $t'$ accepts, plus the number of orders they both accept (all of this divided by $|\Ord|$). 
As an example, $$\pe[M_i]=\pe\Big[\sum_{j\in [n]\setminus i}x_{j,i}-1\Big]=\sum_{j\in [n]\setminus i}\pe[x_{j,i}]-\pe[1]=\sum_{j\in [n]\setminus i}\frac{1}{2}-1=\frac{n-1}{2}-1=\frac{n-3}{2},$$
as exactly half of the total orders are such that $j\prec i$. 

When taking the pseudo-expectation of a sum of terms (or when considering a DNF), we will refer to total orders accepted by more than one term of the sum (or the DNF) as being \emph{over-counted}. 
If a given order is accepted by $c > 1$ terms, we say that this order
contributes $c-1$ \emph{extra occurrences}.
In particular, in the above example, a total order which has $i$ as the 10th element of its order will be accepted by 9 terms, and is thus over-counted. Moreover, 8 of those 9 terms provide extra occurrences of the order.

We will show that our function $\pe$ gives a degree-$d$ $\Sigma_2$-pseudo-expectation for the set of axioms of LOP for any $d \leq \frac{n}{400} -1$.
We proceed by showing that it satisfies the conditions of \hyperref[def:pseudoexp]{Definition~\ref*{def:pseudoexp}}.\\
\mbox{}\\
\noindent \emph{Condition 1.} Note that the constant $1$ is treated as a monomial such that $1(x)=1$ for every $x$, and thus $\pe[1]=1$. \\
\mbox{}\\
\noindent \emph{Condition 2.} 
Notice that conjuncts containing $x_{i,i}$ 
are never satisfied by any total order. We can thus always delete those conjuncts without changing the pseudo-expectation. On the other hand, the negated variables $x_{i,i}$ are always true in any total order
by the axioms $R_i$, thus we can simply remove (or replace by constant 1) any 
negated $x_{i,i}$ without changing the pseudo-expectation.

Additionaly we can rewrite any conjunct $t$ over the LOP variables
$x_{i,j}$ for $i \neq j$ to an equivalent conjunct $t'$ without using negations, simply replacing
any negated variable $x_{i,j}$ by the variable $x_{j,i}$. 
To see this, notice that for any total order on LOP variables, if $i\neq j$, $x_{j,i}$ holds iff $x_{i,j}$ does not hold, and so the total orders which satisfy $t$ are exactly those that satisfy $t’$. Thus, $\pe[t]=\pe[t’]$. 
Observe that the polynomial corresponding to $t'$ is simply a monomial, and
$\pe[t] = \pe[t'] \geq 0$ by definition. Since any conical junta is a sum of conjuncts,
this implies that $\pe[{\cal J}] \geq 0$ for any conical junta ${\cal J}$.\\
\mbox{}\\
\noindent \emph{Condition 3.} It remains to show that condition 3 holds for our function $\pe$.
Our goal is thus to show that for all conical juntas ${\cal J}$, and any weakening $D'$ of any axiom of $\LOP$, $\pe[D'{\cal J}]\geq 0$, if $\deg(D'{\cal J})$ is small enough. 
First we observe that this easily follows for the axioms other than type $M_i$. We will handle axioms of type $M_i$ in \hyperref[sec:m1]{Section~\ref*{sec:m1}}.

 The axioms of LOP that are not of type $M_i$ must hold for every total order. 
 For any such axiom $D = \bigvee_{t \in D} t$ we have 
 $\cup_{t \in D} \{x \in \Ord | t(x)\} = \Ord$. Furthermore, this holds for any weakening $W$ of $D$, as $W$ must be satisfied by at least the same assignments as $D$. Hence, for any conical junta $\cal J$, we can deduce that 
$$\pe[W{\cal J}]=\left( \sum_{t\in W}\frac{|\{x\in\Ord\mid t(x){\cal J}(x)=1\}|}{|\Ord|}\right)
-\frac{|\{x\in\Ord\mid {\cal J}(x)=1\}|}{|\Ord|}\geq0,$$ 
as every order $x$ such that ${\cal J}(x)=1$ must satisfy at least one term $t$ in $W$. This proves that the pseudo-expectation is non-negative for any conical junta and any weakening of an axiom that is not of type $M_i$, regardless of the degree.

\subsection{The No-Minimal-Element Axioms}\label{sec:m1}
By symmetry, it suffices to consider weakenings of $M_1$,
weakenings of $M_i$ can be handled analogously.
It thus suffices, without loss of generality, to prove that for every low-degree conical junta ${\cal J}$, and every low-degree weakening $W$ of $M_1$, $\pe[W{\cal J}]\geq 0$. The main difficulty is to handle the diversity of the weakenings of $M_1$.

We begin by fixing some notation. It will be convenient to go back and forth between total orders and canonical terms representing them. For a set $S \subseteq [n]$ with $|S| \geq 2$, consider an ordering $\pi:[|S|] \rightarrow S$ of its elements. We will denote by $[\![\pi(1) \ldots \pi(|S|)]\!]$, or $[\![S]\!]_\pi$, the term $x_{\pi(1), \pi(2)} x_{\pi(2),\pi(3)}\cdots x_{\pi(|S|-1), \pi(|S|)}$.

We say that a term \emph{mentions $i$} if it contains one of the variables $x_{i,j}$ or $x_{j,i}$ for some $j\neq i$. For example, $x_{1,2}x_{2,3}$ mentions $\{1,2,3\}$. 
We refer to the set of elements that a term mentions as its \emph{support}
and the number of such elements as the \emph{support size} of the term.

To simplify our argument, we will first show that we can assume that each DNF 
has several convenient properties.
\begin{definition}
We say that a DNF $D$ is \emph{normalized} if 
\begin{enumerate}
\item Every term $t$ of $D$ is of the form $t=[\![S ]\!]_\pi$ for some set $S$ which includes $1$, and

\item All terms have the same support size, which we refer to as the \emph{support size of $D$}.
\end{enumerate}
\end{definition}

\begin{lemma}
\phantomsection\label{lem:normalization}

	Let $W$ be a DNF of degree $d$. 
    There exists a normalized DNF $N$ of support size $\leq 2d+1$ such that 
    $W$ accepts exactly the same total orders as $N$ does,
    and $\pe[W]=\pe[N]$. 
\end{lemma}
\begin{proof}
	Let $t \not \equiv 0$ be one of the terms of $W$. We first transform $t$ into a term $t’$ with no negated literals as follows: we replace each negated literal $\neg x_{i,j}$ in $t$ with $x_{j,i}$, noting that for any total order, $x_{j,i}$ holds iff $x_{i,j}$ does not, and so the total orders satisfying $t$ are exactly those satisfying $t’$. Hence, $\pe[t]=\pe[t’]$.

Next, we replace $t$ by a collection of terms of the form $[\![S]\!]_\pi$ for some $S \subseteq [n]$ containing $1$ (defined below), and a collection of orderings $\pi$ on $S$. 
Let $T$ be the support of $t$ and note that $|T| \leq 2d$ as each variable in $t$ can mention at most two elements. We modify $T$ to define $S$ as follows:
Include the element $1$  if it was not included yet in $T$, and 
pad $T$ with additional elements as necessary, so that the resulting set $S$ 
has size $|S|=2d +1$.
We use the resulting set $S$ as the support for the collection of terms replacing $t$.

Let $\Pi$ be the set of total orders on $S$ that are consistent with $t$. We claim that 
for every total order $z$, $t(z)\Leftrightarrow \bigvee_{\pi\in\Pi} \bra{S}[\pi](z)$, 
  and $\pe[t]=\pe[ \sum_{\pi\in\Pi} \bra{S}[\pi]]$. In other words, $\bigvee_{\pi\in\Pi} \bra{S}[\pi]$ accepts exactly the same total orders $t$ does, 
and each total order accepted by $t$ is consistent with exactly one $\pi \in \Pi$
(thus total orders accepted by $t$ are accepted by a unique term in $\bigvee_{\pi\in\Pi} \bra{S}[\pi]$). We set $N$ to be the resulting DNF $\bigvee_{t\in W}\bigvee_{\pi\in\Pi_t}\bra{S_t}[\pi]$.
\end{proof}

Note that the transformation in the proof of \hyperref[lem:normalization]{Lemma~\ref*{lem:normalization}} can highly increase the number of terms in $W$, but it does not change its pseudo-expectation,
and it at most doubles the degree: 
Notice that for a normalized DNF with support size $k$ its degree is exactly
$k-1$.
The support size of the resulting DNF $N$
is at most $2d +1$, thus the degree of $N$ is at most $2d$.

Suppose that $M_1$ has a weakening $W$ of degree $d$ such that $\pe[W]<0$. This lemma implies that there is also a normalized weakening $N$ of $M_1$ with $\pe[N]<0$ and of not much larger degree. 

\begin{lemma}\label{count}
For any normalized DNF $N$ of support size $k$, every term $t$ of $N$ has $\pe[t] = 1/k!$. 
\end{lemma}

\begin{proof}
Let $t=\bra{T}[\pi]$ be one of the terms of $N$.
As a reminder, $\pe[t]=|\{x\in \Ord\mid t(x)=1\}|/|\Ord|$ where $\Ord$ is the set of total orders. Since $t=\bra{T}[\pi]$, $\pe[t]=|\{x\in \Ord \mid x\rest T=\pi\}|/|\Ord|$. To count the number of orders consistent with $t$, one can simply select the positions of the elements in $T$, set them in the right order, and then fill in the rest any possible way, thus,
    \begin{align*}
      |\{x\in \Ord \mid x\rest T=\pi\}|&={n\choose k}(n-k)!=\frac{n!}{k!}. \qedhere
    \end{align*}
\end{proof}

\begin{lemma}\label{prop:everythingForPE}
  Let $N$ be normalized DNF with support size $k$ that is a weakening of $M_1$,
and let $R$ be the set of total orders rejected by $N$. 
  Then the following holds: 
  \begin{enumerate}
    
    \item $\pe[N]\geq 0$ if and only if $N$ has at least $k!$ terms, 
    \item Let $S_i$ be the set of total orders accepted by exactly $i$ terms in $N$, then $\pe[N]\geq 0$ if and only if $|R|\leq \sum_{i=2}^{\infty}(i-1)|S_i|$.
  \end{enumerate}
\end{lemma}

We remark that the equivalence with the covering problem presented in the introduction comes from the first point in this lemma. Finding a covering of $\Ord^{*1}$ with less than $d!$ sets (that correspond here to normalized terms) is equivalent to finding a normalized weakening $N$ with support size $d$ and negative pseudo-expectation.

\begin{proof}
  For the first point, using \autoref{count}, $$\pe[N]=\pe \bigg[\sum_{\bra{T}[\pi]\in N}\bra{T}[\pi]-1 \bigg]=\sum_{\bra{T}[\pi]\in N}\pe[\bra{T}[\pi]]-1=\frac{\#\text{of terms in $N$}}{k!}-1.$$
    This is non-negative if and only if there are at least $k!$ terms in $N$. 
    For the second point, recall that for terms of the form
    $t=\bra{T}[\pi]$, we have $\pe[t]=|\{x\in \Ord \mid x\rest T=\pi\}|/|\Ord|$. Thus
    \begin{align*}
    \pe[N]&=\frac{1}{|\Ord|}\bigg(\sum_{\bra{T}[\pi]\in N}|\{x\in \Ord\mid x\rest T=\pi\}|\bigg) -1\\
    &=\frac{1}{|\Ord|} \bigg( \sum_{x\in \Ord} (|\{\bra{T}[\pi] \in N\mid x\rest T=\pi\}| \bigg) -\frac{|\Ord|}{|\Ord|}\\
    &=\frac{1}{|\Ord|} \bigg(\sum_{i=1}^\infty i |S_i|\bigg) - \frac{1}{|\Ord|}  \bigg( |R|+\sum_{i=1}^\infty |S_i| \bigg)\\
    &=\frac{1}{|\Ord|}\bigg( \sum_{i=2}^{\infty}(i-1)|S_i| - |R|\bigg). \qedhere
    \end{align*}
\end{proof}

The proof of \autoref{thm:LOPnotsRA} uses the second point of this lemma: we show that if the degree of a normalized weakening is low enough, then a large number of orders must be accepted by several terms of $W{\cal J}$.

\subsubsection{A Warmup}\label{sec:warmup}
We begin by considering the case when ${\cal J} =1$.
As a warmup, we show a simplified (but not sufficient) argument,
where we only focus on orders accepted by at least two terms.
Later we present the stronger argument considering orders accepted 
by at least 3 terms which is sufficient to prove our theorem. We do this for the general case of arbitrary conical juntas in \hyperref[sec:general]{Section~\ref*{sec:general}}.

We use the notation $1z$ to denote a total order on $[n]$ where the element $1$ is first and $z$ is a total order of the remaining $n-1$ elements.
The notation $i1z$ is similarly defined: it denotes orders where $i$ is first, $1$
is second, and $z$ is a total order of the remaining $n-2$ elements.

\begin{lemma}\label{lem:accept}
For a normalized DNF,
\begin{enumerate}
\item Total orders starting with 1 can only be accepted by terms $\bra{T}[\pi]$ with $\pi(1)=1$. 
\item Total orders with first element $i$ and second element 1 can only be accepted by terms $\bra{T}[\pi]$ with either $\pi(1)=1$ and $i\notin T$, or $\pi(1)=i$ and 
$\pi(2)=1$. 
\end{enumerate}    
\end{lemma}
\begin{proof}
For the first point, since $1$ is in the support $T$ for each term in a normalized DNF, 
if $\pi(1)\neq 1$ then $\pi$ is not equal to $1z\rest T$ for any $1z\in \Ord$.
For the second, consider the following two cases:
If $\pi(1)=1$ and $i\in T$, then $i1z\rest T \neq \pi$ for any $i1z\in \Ord$. If $\pi(2)=1$ and $\pi(1)\neq i$, then $i1z\rest T\neq \pi$. Otherwise, if $\pi^{-1}(1)>2$, then $i1z\rest T\neq \pi$.
\end{proof}

The following definitions help to describe our combinatorial argument.
To simplify notation, we denote $z\setm \{i\}$ by $z\setm i$, and
$[n]\setm \{1\}$ by $[n]\setm 1$.

\begin{definition}
Let $z$ be an order on $[n]\setm 1$. We say that a term $\bra{T}[\pi]$ is
{\em $z$-good for $i$} if 
$i \in T$, $1 \in T$, $\pi(1)=i, \pi(2)=1$, and $i1(z\setm i)\rest T=\pi$. That is, the order $\pi$ on $T$ places $i$ first, $1$ second, and is 
consistent with the total order $i1(z\setm i)$.
\end{definition}

Notice that if a term $\bra{T}[\pi]$ is
$z$-good for $i$ then it accepts the order $i1(z\setm i)$.
Moreover, it also accepts the orders $ij1(z\setm \{i,j\})$ and $ji1(z\setm \{i,j\})$
for any $j$ that is not in its support.

\begin{definition} Let $z$ be an order on $[n]\setm 1$.
 We call an unordered pair of distinct elements $(i,j)\in ([n]\setm1)^2$ a \emph{hitting pair} for $z$ if the following conditions hold:
  \begin{itemize}
    \item $N$ has a term $\bra{T_1}[\pi_1]$ such that $\pi_1(1)=i, \pi_1(2)=1, j\notin T_1$ and $i1(z\setm i)\rest T_1=\pi_1$.
    \item $N$ has a term $\bra{T_2}[\pi_2]$ such that $\pi_2(1)=j, \pi_2(2)=1, i\notin T_2$ and $j1(z\setm j)\rest T_2=\pi_2$.
  \end{itemize}    
\end{definition}    
In other words, a pair
   $(i,j)$ is a hitting pair for $z$,
  if there is some term that is $z$-good for $i$ but its support does not contain $j$,
  and there is some term that is $z$-good for $j$ but its support does not contain $i$.

 Notice that if $(i,j)$ is a hitting pair for $z$, then \emph{both} $ij1(z\setm \{i,j\})$ and $ji1(z\setm \{i,j\})$ are accepted by \emph{both} $\bra{T_1}[\pi_1]$ and $\bra{T_2}[\pi_2]$. Therefore, we can get a lower bound on the 
  number of orders accepted by several terms of $N$ by getting a lower bound on the number of hitting pairs.

\begin{lemma}\label{lem:pairs}
    Let $N$ be a normalized DNF with support size $k$ that is a weakening of $M_1$. Let $R\subseteq \Ord$ be the set of orders rejected by $N$. If $k\leq n/10$, there are at least $4|R|/5$ total orders accepted by more than one term of $N$.
\end{lemma}
\begin{proof}
    If $R=\emptyset$, we are done, so suppose otherwise.
Recall that we assume for normalized DNFs that each term contains the element $1$,
so we do not explicitly state this condition.

Given an order $z$ on $[n]\setm 1$, we define a matrix $H_z$ as follows.
  Let $H_z$ be an $(n-1)\times(n-1)$ matrix with rows and columns indexed by $i,j \in [n] \setminus 1$, respectively,  where $H_z(i,j)=1$ iff $i \neq j$, and
  $N$ has a term $\bra{T}[\pi]$ such that $\pi(1)=i$, $\pi(2)=1$, $j\notin T$, and $i1(z\setm i)\rest T=\pi$. That is, the $(i,j)$th entry is 1 iff there is a term
  that is $z$-good for $i$ but its support does not contain $j$.
  Thus, $(i,j)$ is a hitting pair iff both $H_z(i,j)=1$ and $H_z(j,i)=1$. 

Let $1z\in R$ be an order rejected by $N$. Note that since $N$ is a weakening of $M_1$, every order it rejects must start with 1. Moreover, for every $i\in[n]\setm 1$, $N$ must accept the total order $i1(z\setm i)$, as it is accepted by $M_1$. 

If $N$ had any terms $\bra{T}[\pi]$ with $\pi(1)=1$ and $1z\rest T=\pi$, $1z$ would not be rejected. Thus, by point 2 of 
    \autoref{lem:accept}, for each $i\in[n]\setm 1$, the order $i1(z\setm i)$ must be accepted by a term $\bra{T}[\pi]$ of $N$ with $\pi(1)=i, \pi(2)=1$.
    
That is, for each $i\in[n]\setm 1$, and any $z$ such that $1z$ is rejected by $N$,
there is a term of $N$ that is $z$-good for $i$.
Since $N$ is normalized, the support of each term has size $k$, and so row $i$ of $H_z$ must have at least $n-k$ entries with value $1$ (there are $n - k$ elements $j$ not included in the support of any particular term). The matrix thus contains at least $(n-1)(n-k)$ entries which are $1$.

Next, observe that the largest number of $1$-entries in a matrix without
a hitting pair is at most ${n-1 \choose 2}$: by definition, any such matrix has a $1$ in at most one of the entries $H_z(i,j)$ or $H_z(j,i)$, for each unordered pair $(i,j)$. Moreover, the following holds.
\begin{observation}\label{hitting-count}
Any $(n-1) \times (n-1)$ $\{0,1\}$-matrix with $q > {n-1 \choose 2}$ $1$-entries outside of the diagonal must
contain at least $q - {n-1 \choose 2}$ hitting pairs.
\end{observation}
\begin{proof}
Let $H$ be an arbitrary $(n-1) \times (n-1)$ $\{0,1\}$-matrix with $q > {n-1 \choose 2}$ $1$-entries, none of which are on the diagonal.
Since the number of $1$-entries is more than ${n-1 \choose 2}$, $H$ must contain
at least one hitting pair $(i,j)$. Replace one of the entries corresponding to
this pair (either the entry $H(i,j)$ or the entry $H(j,i)$) by $0$.
We can repeat this step at least $q - {n-1 \choose 2}$ times, accounting
for at least $q - {n-1 \choose 2}$ hitting pairs.
\end{proof}
Thus, in the matrix $H_z$ there are at least $(n-1)(n-k)-{n-1\choose 2}$ hitting pairs for $z$.
  This holds for each $z$ such that $1z\in R$. 

  Each of the hitting pairs in $H_z$ causes (at least) two orders to be accepted twice,
  in particular, as discussed above, if $(i,j)$ is a hitting pair for $z$, then \emph{both} $ij1(z\setm \{i,j\})$ and $ji1(z\setm \{i,j\})$ are accepted twice.
  Notice that the  total orders we count twice this way due to $(i,j)$ 
  are the same for $z$ and $z'$ 
  iff $z\setm\{i,j\}=z'\setm\{i,j\}$.
  For any given $z$, the number of $z'$ such that $z\setm\{i,j\}=z'\setm\{i,j\}$ is
  at most $(n-1)(n-2)$. 

  Thus, we can conclude
  that there must be at least $$\frac{2\left((n-1)(n-k)-{n-1\choose 2}\right)}{(n-1)(n-2)}|R|$$
  orders accepted by at least two terms of $N$. This value is always smaller than $|R|$. However, if $k\leq n/10$, it is at least $4|R|/5$.
\end{proof}

\subsubsection{Normalization with Respect to Juntas}

The above argument shows that focusing on hitting pairs is almost enough to conclude that $\pe[W]\geq 0$ for weakenings $W$, even if they have a fairly high degree. Our final
proof for the general case is similar to this argument, but focuses on hitting triples. Before stating and proving this result, we need to deal with conical juntas. 
We extend the normalization done in \hyperref[lem:normalization]{Lemma~\ref{lem:normalization}} to handle the additional inclusion of a conical junta. This is captured by the following lemmas.

\begin{lemma}\label{lem:JuntaNormalization1}
    Let the DNF $W$ be a weakening of $M_1$ and ${\cal J}$ be a conical junta. If $\pe[W{\cal J}]<0$ there exists $N$ and $t$ such that 
    \begin{enumerate}
        \item $N$ is the normalized form of $W$ given by \hyperref[lem:normalization]{Lemma~\ref*{lem:normalization}},
        \item $t=\bra{T}[\pi]$ such that  $1\in T$, $\pi(1)=1$, and the support size of $t$ is $\leq 2\deg({\cal J})+1$
        \item $\deg(Nt)\leq 2\deg(W) + 2\deg({\cal J})$, 
        \item $\pe[Nt]<0$. 
    \end{enumerate}
\end{lemma}

\begin{proof}
Let ${\cal J}=\sum_{t\in {\cal J}} t $ and suppose that $0>\pe[W{\cal J}] = \sum_{t\in {\cal J}}\pe[Wt]$, where the equality follows by linearity. As the sum is negative, there must be at least one term $t^*$ such that  $\pe[Wt^*]<0$. 
    Let $T$ be the support of $t^*$. Note that without loss of generality we can assume that $1$ belongs to $T$: If $1\notin T$ then let $\Pi$ be the set of permutations on $T\cup\{1\}$ consistent with $t^*$. As shown in the proof of \hyperref[lem:normalization]{Lemma~\ref*{lem:normalization}}, $t^*$ is equivalent
    to $\bigvee_{\pi\in \Pi}\bra{T\cup \{1\}}[\pi]$, and $\pe[t^*]=\pe[\sum_{\pi\in \Pi}\bra{T\cup \{1\}}[\pi]]$. In particular, every order accepted by $t^*$ is accepted by $\bra{T\cup\{1\}}[\pi]$ for exactly one $\pi\in \Pi$. Thus, $\pe[Wt^*] = \sum_{\pi\in\Pi}\pe[W\bra{T\cup \{1\}}[\pi]]$, and if $\pe[Wt^*]<0$ then there must exist $\pi\in\Pi$ such that $\pe[W\bra{T\cup \{1\}}[\pi]]<0$. 
    
    Suppose that $\pi(1)\neq 1$, then $W \wedge \bra{T}[\pi]$ does not accept any order that starts with $1$, since $\bra{T}[\pi]$ does not. As $W$ is a weakening of $M_1$, it must accept every order that does not start with $1$. 
    Thus, $W \wedge \bra{T}[\pi]$ is equivalent to $\bra{T}[\pi]$
    and $\sum_{w\in W}\pe[w\bra{T}[\pi]]\geq  \pe[\bra{T}[\pi]]$, as every order accepted by $\bra{T}[\pi]$ must be accepted by
    at least one term $w$ of $W$.
    This implies that
    $$\pe\big[W\bra{T}[\pi]\big]=\pe \bigg[\Big(\sum_{w\in W}w -1\Big)\bra{T}[\pi]\bigg] = \sum_{w\in W}\pe[w\bra{T}[\pi]] - \pe[\bra{T}[\pi]] \geq \pe[\bra{T}[\pi]]-\pe[\bra{T}[\pi]]=0.$$
    Combining  this with the previous argument, we get that if there is some conical junta such that $\pe[W{\cal J}]<0$ then there is also a term $t=\bra{T}[\pi]$, with $1\in T$ and $\pi(1)=1$, such that $\pe[Wt]<0$. Note that the support size of $t=\bra{T}[\pi]$ is at most $2\deg({\cal J})+1$, and its degree 
    is at most $2\deg({\cal J})$.

    Next, we apply \hyperref[lem:normalization]{Lemma~\ref*{lem:normalization}} to the DNF $W$ to obtain a 
    normalized DNF $N$. Recall that the degree of $N$ will be at most $2\deg(W)$.
    Since both the degree of $N$ and $t$ might have at most doubled, and 
    $\deg(Nt)\leq\deg(N)+\deg(t)$, we get that 
    $\deg(Nt)\leq 2 \deg(W) + 2 \deg({\cal J})$, concluding (3).

We now argue that $\pe[Wt] < 0$ implies that   $\pe[Nt] < 0$. Expanding $Nt$, we have 
\[ \pe[Nt] = \pe \Big[ \Big( \sum_{v \in N} v-1 \Big) t \Big] = \sum_{v \in N}\pe[vt]-\pe[t].\]
 Since $N$ is the normalization of $W$, the terms of $N$ can be partitioned into groups, one for each term of $W$, such that the term of $W$ accepts a given order iff a unique term in its corresponding group accepts that order, and the term of $W$ rejects a given order iff all terms in the corresponding group reject it. Thus, $\sum_{v\in N}\pe[vt]=\sum_{w\in W}\pe[wt]$, and so $\pe[Nt] = \pe[Wt]<0$.
 \end{proof}

 Note that in the above lemma, both $N$ and $t$ are normalized, but the DNF 
 corresponding to $N \wedge t$ (obtained by taking the conjunction of each term of $N$ with $t$) is not. We address this issue next.
 
      \begin{lemma}\label{lem:JuntaNormalization2}
  Let $t=\bra{T}[\pi^*]$ be a term of support size $\ell$ such that
  $1\in T$ and $\pi^*(1)=1$ , and let $N$ be a normalized DNF of support size $h$ that is a weakening of $M_1$.       
        Then there is a DNF $N'$ such that 
    \begin{enumerate}
        \item For every total order $x$, $N'(x)\Leftrightarrow N(x)\wedge t(x)$. Hence, $N'$ is a weakening of $M_1\wedge t$,
        \item $N'$ is normalized with support size at most $\ell +h-1$,
        and degree at most $\ell + h -2 \leq \deg(N) + \deg(t)$,
        \item $\pe[Nt]=\pe[N'+1-t]$.
    \end{enumerate}
\end{lemma}

Coming back to our analogy with the covering problem presented in the introduction, proving that $\pe[N'+1-t]$ is non-negative for all such $N'$ of small degree (for a fixed $t$), is equivalent to saying that, even if we restrict the sets $\Ord$ and $\Ord^{*1}$ to their orders consistent with $t$, it is impossible to get a covering of $(\Ord\wedge t)^{*1}$ with fewer sets than the number needed to cover $\Ord\wedge t$, if the support size of the sets is small. 

\begin{proof}[Proof of \autoref{lem:JuntaNormalization2}]
  We construct $N'$ by ``merging" terms of $N$ and $t$. Formally, let $\bra{S}[\sigma]$ be a term of $N$, and let $\Pi_{S,\sigma}$ be the set of permutations on $S\cup T$ consistent with both $\sigma$ and $\pi^*$; note that this set might be empty. Define $N':=\bigvee_{\bra{S}[\sigma]\in N}\bigvee_{\pi\in \Pi_{S,\sigma}}\bra{S\cup T}_{\pi}$ and observe that  this DNF accepts the same total orders $Nt$ does. Indeed, an order is consistent with a term $\bra{S}[\sigma]$ of $N$ and $t$ iff there is an ordering in $\Pi_{S,\sigma}$ consistent with it. Every term in $N'$ is normalized, but since for different terms $\bra{S}[\sigma], \bra{S'}[\sigma]$ of $N$, the sizes of $S\cup T$ and $S'\cup T$ might be different, terms in $N'$ might be of different sizes. As in the previous arguments, one can simply add elements to the support of smaller terms to get them to a fixed size. The maximal support size of a term in $N'$ is the maximal size of $S\cup T$, where $S$ is the support of some term in $N$. Since $N$ and $t$ are both normalized, and since they both contain $1$, this value is at most $\ell + h-1$. This also implies that the degree $\deg(N') \leq \ell + h -2 \leq \deg(N) + \deg(t)$.
    
    It remains to prove that $\pe[Nt]=\pe[N'+1-t]$. For a single term $\bra{S}[\sigma]$ of $N$ it holds that $\pe[\bra{S}[\sigma]t]=\pe~[\sum_{\pi\in \Pi_{S,\sigma}}\bra{S\cup T}_{\pi}]$, as every total order accepted by $\bra{S}[\sigma]t$ is consistent with a unique $\pi\in \Pi_{S,\sigma}$. Thus, 
    \begin{align*}
        \pe[Nt]&=\pe \Big[\sum_{\bra{S}[\sigma]\in N}\bra{S}[\sigma]t- t \Big]=\sum_{\bra{S}[\sigma]\in N}\pe\big[\bra{S}[\sigma]t\big]- \pe[t]\\
        &=\sum_{\bra{S}[\sigma]\in N}\pe\bigg[\sum_{\pi\in \Pi_{S,\sigma}}\bra{S\cup T}_{\pi}\bigg]- \pe[t]=\pe\bigg[\sum_{\bra{S}[\sigma]\in N}\sum_{\pi\in \Pi_{S,\sigma}}\bra{S\cup T}_{\pi}\bigg]- \pe[t]\\
        &=\pe[N'+1]- \pe[t]=\pe[N'+1-t]. \qedhere
    \end{align*}
\end{proof}

\subsubsection{The General Case}\label{sec:general}
We are now ready to state and prove our main theorem about weakenings of $M_1$. 
\begin{theorem} \label{thm:PEwJunta}
    Let $t=\bra{T^*}[\pi^*]$ be a term of support size $\ell$ such that $\pi^*(1)=1$ , and let $N$ be a normalized DNF of support size $h$ that is a weakening of $M_1$. 
    If $\ell+h\leq n/100$, then $\pe[Nt]\geq 0$. 
\end{theorem}

In \hyperref[sec:warmup]{Section~\ref*{sec:warmup}} we presented a warmup argument for the case when the conical junta, ${\cal J}$, was assumed to be the trivial junta ${\cal J} =1$.
This simpler argument (considering only "pairs") 
is not sufficient to get a good bound, 
even for the case when ${\cal J} =1$,
but it illustrates our approach well. 
To get a good enough bound, we extend this argument to triples.
To prove \autoref{thm:PEwJunta},  we  allow ${\cal J}$ to be an arbitrary conical junta. This also includes the case when $\cal J$ could be the trivial junta ${\cal J}=1$, in which case we define its support to be  $T^*=\{1\}$.

\begin{proof}[Proof of \autoref{thm:PEwJunta}]
    Let $N'$ be the weakening of $M_1 \wedge t$ given by \autoref{lem:JuntaNormalization2}; that is $\pe[Nt]=\pe[N'+1-t]$. We will prove that $\pe[N'+1-t]\geq 0$. 
    Let $\ell+g$ be the support size of $N'$, 
    and recall that this is at most $\ell+h-1\leq n/100-1$. Let $R$ be the set of total orders that are consistent with $t$ and are rejected by $N'$, and hence also  by $N\wedge t$. Recall also that by \autoref{lem:JuntaNormalization2} $N'$ only 
    accepts orders consistent with $t$.
    We will prove that the over-counting in the pseudo-expectation that is due to orders that are accepted by several terms of $N'$ is sufficiently large compared to the number of orders rejected by $N'$ that are consistent with $t$. A sufficiently large "over-count" will be enough to show that $\pe[N'+1-t]\geq 0$,
    since  $\pe[N'+1-t]$ is exactly the sum of the number of orders that each term of $N'$ accepts minus the number of orders consistent with $t$, divided by $n!$. 
    $\pe[N'+1-t]$ is thus non-negative, if and only if 
    the number of extra occurrences of orders accepted by more than one term of $N'$ is at least $|R|$ (Recall that $|R|$ is the number of orders consistent with $t$ that are not accepted by any term of $N'$).

    Given an order $z$ on $[n]\setm 1$,
    we call an unordered triple of distinct elements $(i,j,k)\in([n]\setm T^*)^3$ a \emph{hitting triple for $z$} if the following hold: 
    \begin{enumerate}
        \item $N'$ has a term $\bra{T_1}[\pi_1]$ such that $\pi_1(1)=i, \pi_1(2)=1, j,k \notin T_1$ and $i1(z\setm i)\rest T_1=\pi_1$.
        \item $N'$ has a term $\bra{T_2}[\pi_2]$ such that $\pi_2(1)=j, \pi_2(2)=1, i,k \notin T_2$ and $j1(z\setm j)\rest T_2=\pi_2$.
        \item $N'$ has a term $\bra{T_3}[\pi_3]$ such that $\pi_3(1)=k, \pi_3(2)=1, i,j\notin T_3$ and $k1(z\setm k)\rest T_3=\pi_3$.
  \end{enumerate}
  In other words, 
      $(i,j,k)$ is a hitting triple for $z$,
  if there is some term that is $z$-good for $i$ 
  but $j$ and $k$ are missing from its support,
  there is some term that is $z$-good for $j$ 
  but $i$ and $k$ are missing from its support,
  and there is some term that is $z$-good for $k$ 
  but $i$ and $j$ are missing from its support.
    Analogous to hitting pairs, the presence of hitting triples implies that there are orders
    accepted by several terms.
    Indeed, if $(i,j,k)$ is a hitting triple for $z$, then for all six permutations $\sigma$ of $\{i,j,k\}$, the order $\sigma(1)\sigma(2)\sigma(3)1(z\setm\{i,j,k\})$ is accepted by at least three terms. 
    
    Given an order $z$ on $[n]\setm 1$, define an array $H_z:([n] \setminus T^*)^3 \rightarrow \{0,1\}$, such that $H_z(i,j,k)=1$ iff $i,j,k$ are distinct and not in $T^*$, 
    and $N'$ has a term that is $z$-good for $i$ 
  but $j$ and $k$ are missing from its support
(that is $N'$ has a term $\bra{S}[\pi]$ with $\pi(1)=i, \pi(2)=1$, $j,k\notin S$ and $i1(z\setm i)\rest S=\pi$). 
Note that by this definition,  $H_z(i,j,k)=H_z(i,k,j)$. Furthermore, notice that
$(i,j,k)$ is a hitting triple for $z$ iff for all six permutations $\pi$ of $\{i,j,k\}$, $H_z(\pi(1),\pi(2),\pi(3))=1$. 
In other words, 1-entries now come in pairs in the array, and one needs six 1-entries to get a hitting-triple. Thus, there can be at most $4{n-\ell \choose 3}$ 1-entries in an array $H_z$ without any hitting triples. 
    The $-\ell$ comes from the fact that we do not consider elements of $T^*$ in the triples.
    Similarly to \autoref{hitting-count}, we get the following.
    \begin{observation}\label{triple-count}
    Let $H$ be an $(n-\ell) \times (n-\ell) \times (n-\ell)$ \{0,1\}-array with $q > 4{n-\ell \choose 3}$ $1$-entries, where all $1$ entries are  in positions 
    with distinct coordinates $i,j,k$, such that $H(i,j,k) = H(i,k,j)$ for all 
    $i,j,k$. Then $H$
    must contain at least $\frac{1}{2}\left(q - 4{n-\ell \choose 3}\right)$ hitting triples. 
    \end{observation}
    \begin{proof}
    
    Let $H$ be an  $(n-\ell) \times (n-\ell) \times (n-\ell)$ \{0,1\}-array with $q > 4{n-\ell \choose 3}$ $1$-entries, where all $1$-entries are in positions 
    with distinct coordinates $i,j,k$, such that $H(i,j,k) = H(i,k,j)$ for all 
    $i,j,k$. Since the $1$-entries come in pairs, we have $q \geq 4{n-\ell \choose 3} +2$ 
    entries that are $1$, which implies that for at least one triple with distinct $i,j,k$,
    all 6 corresponding entries must be $1$. Thus $H$ contains at least one hitting triple.
    Replace one pair of entries corresponding to this triple, say $H(i,j,k)$ and $H(i,k,j)$, by $0$.
    We can repeat this step at least $\frac{1}{2}\left(q - 4{n-\ell \choose 3}\right)$ times, accounting
    for at least $\frac{1}{2}\left(q - 4{n-\ell \choose 3}\right)$ hitting triples.
    \end{proof}
    
    Let $1z\in R$ be a rejected order that is consistent with $t$. 
    By a similar argument to the proof of \autoref{lem:pairs}, for every  $i\in[n]\setm T^*$, 
    there is a term of $N'$ that is $z$-good for $i$
    (and thus this term accepts the order $i1(z\setm i)$).
    This only holds for $i$ that are not in $T^*$ (the support of $t$), since $N'$ only accepts total orders that are accepted by $t$, that is only total orders that are consistent with $\pi^*$ on $T^*$, and moving any element in $T^*$ before $1$ would contradict $\pi^*$.
     
    Since for all $i\notin T^*$, the order $i1(z\setm i)$ is accepted by a term of support size $\ell+ g$, with $\pi(1)=i,\pi(2)=1$, there must be at least
    $(n-\ell-g)(n-\ell -g-1)$ 1-entries in $H_z$ at positions starting with $i$. Hence, there are at least $\frac{1}{2}\left((n-\ell )(n-\ell-g)(n-\ell -g-1)-4{n-\ell\choose 3}\right)$ hitting triples for $z$. Each hitting triple gives (at least) six orders accepted by (at least) three terms: specifically, as noted earlier,
    for all six permutations $\sigma$ of $\{i,j,k\}$, the order 
    $\sigma(1)\sigma(2)\sigma(3)1(z\setm\{i,j,k\})$ is accepted by at least three terms. 
    Thus each hitting triple gives (at least) 12 total extra occurences of 
    over-counted orders.
    By a similar reasoning to the proof of \autoref{lem:pairs}, each of these 
    specific orders could come from hitting triples of at most $(n-1)(n-2)(n-3)$ different $z$-s, since for any given  order $z$ on $[n]\setm 1$ consistent with $t$, the number of orders $z'$ on $[n]\setm 1$ consistent with $t$ such that $z\setm\{i,j,k\}=z'\setm\{i,j,k\}$ is
  at most $(n - 1)(n-2)(n-3)$.

    Thus, we get that there are at least 
    $$\frac{6\left((n-\ell)(n-\ell-g)(n-\ell-g-1)-4{n-\ell\choose 3}\right)}{(n- 1)(n- 2)(n- 3)}|R|$$
    total extra occurrences of over-counted orders. 
    
    Plugging in $\ell +g+1\leq n/100$, we get 
    \begin{align*}
        &\frac{6\left((n-\ell)(n- \ell-g)(n- \ell-g-1)-4{n- \ell\choose 3}\right)}{(n- 1)(n- 2)(n- 3)}|R|\\
        &=\left(\frac{6(n-\ell)(n- \ell-g)(n- \ell-g-1)}{(n- 1)(n- 2)(n- 3)}-\frac{24{n- \ell\choose 3}}{(n- 1)(n- 2)(n- 3)}\right)|R|\\
        &\geq \left(\frac{6\left((99n/100)^3\right)}{n^3}-\frac{4 (n-\ell)(n-\ell-1)(n-\ell-2)}{(n- 1)(n- 2)(n- 3)}\right)|R|\\
        &\geq \left(\frac{6\left((99n/100)^3\right)}{n^3}-4\right)|R|=\left(6\frac{99^3}{100^3}-4\right)|R|\geq |R| \qedhere
    \end{align*}

\end{proof}

We can now conclude the proof of the main theorem of this section.

\begin{proof}[Proof of {\hyperref[thm:PE]{Theorem~\ref*{thm:PE}}}] 
As we argued above, it suffices to consider weakenings of $M_1$.
    By \autoref{lem:JuntaNormalization1}, 
    if  $\pe[W{\cal J}]<0$ for a weakening $W$ of $M_1$ and conical junta ${\cal J}$, both of degree at most $d$,
    then there is a normalized weakening $N$ of support size $h \leq 2d +1$ and a normalized term $t$ of support size
    $\ell \leq 2d +1$ such that $\pe[Nt] < 0$.
    Notice that if $d \leq (\frac{n}{400} -1)$, then $\ell + h \leq 4d + 2 \leq n/100$.
    By \autoref{thm:PEwJunta}, this cannot happen. 
    Thus, there is a $\Sigma_2$-pseudo-expectation for $\LOP$ of degree 
    $(\frac{n}{400} -1)$.
\end{proof}

Note that the $n/100$ bound was picked for convenience, one could get a better bound by counting more than the hitting triples, and by being more careful in the computation. 

\section{A Criterion for Non-Reducibility of $\LOP$}
\label{sec:furtherSep}

In this section, we show that our pseudo-expectation for $\LOP$ (\hyperref[lem:PESigma2]{Lemma~\ref*{lem:PESigma2}}) implies a general criterion for the non-reducibility of $\LOP$ to other problems in $\TF\Sigma_2^{dt}$.

\begin{corollary}\label{lem:LOP_non_reduc}
  Suppose that $R\in\TF\Sigma_2^{dt}$ and $R$ has a $\polylog(n)$-degree $\SA$-refutation. Then $\LOP\nleq_{dt} R$.
\end{corollary}

Using this, we show that $\LOP \not \in \MINclass^{dt}$.

\begin{proof}[Proof of \autoref{thm:LOPvsLN}]
    The axioms of the formula of $\MIN_n$ are:
    \begin{itemize}
        \item $\bigvee_{i\in[n]}x_i$, giving the inequality $\sum_{i\in[n]}x_i-1\geq 0$; \hfill \emph{(Not all $0$)}
        \item $\overline{x_i}\vee\bigvee_{j<i}x_j$, giving the inequality $(1-x_i)+\sum_{j<i}x_j -1= -x_i+\sum_{j<i}x_j \geq 0$ for each $i\in[n]$. \\ \phantom{a}\hfill \emph{($i$ is the least index taking value $1$)}
    \end{itemize}
    This admits the following simple $\SA$-refutation,
    $$\sum_{i=1}^nx_i-1+\sum_{i=0}^{n-1}2^i\Big(-x_{n-i}+\sum_{j<n-i}x_j\Big)=-1+\sum_{i=1}^n\Big(1+\sum_{j=0}^{n-i-1}2^j-2^{n-i}\Big)x_i=-1.$$
    Although, the refutation is of exponential size, it has constant degree.
    By \autoref{lem:LOP_non_reduc}, $\LOP\notin\MINclass^{dt}$.
\end{proof} 

We now prove \autoref{lem:LOP_non_reduc}. To do so, in \hyperref[lem:factorization]{Lemma~\ref*{lem:factorization}} we show that we can factorize any reduction between $\TF\Sigma_2^{dt}$ problems into two parts (i) a weakening step, akin to $\Sigma_2$-weakening, and (ii) a \emph{counter-example reduction} \cite{KolodziejczykT22,Thapen2024} which takes place in $\TFNP^{dt}$. This elucidates how the two types of reductions between $\TF\Sigma_2^{dt}$ problems that have been considered (formulations and counter-example reductions) relate. Using this, our proof proceeds as follows: 
\begin{enumerate}
    \item Consider an $R$-formulation of $\LOP$ where $R$ admits a low-degree $\SA$-refutation.
    \item Factorize the $R$-formulation into two parts using \hyperref[lem:factorization]{Lemma~\ref*{lem:factorization}}: a weakening from $\LOP$ to a problem $P$, and a counter-example reduction from $P$ to $R$.
    \item Argue that we can derive an $\SA$-refutation of $P$ from an $\SA$-refutation of $R$ while roughly maintaining the degree (\autoref{prop:SA_proof_reduced_problem}).
\end{enumerate}

\subsection{Factoring Reductions}

To improve readability, we will use a slightly different notation for $\TF\Sigma_2^{dt}$ problems in this section.
\begin{notation}
    A $\TF\Sigma_2^{dt}$ search problem $R$ is a sequence of search problems $R_m\subseteq\{0,1\}^m\times \mathcal{O}^R_m\times \mathcal{W}^R_{m}$ for each $m\in\mathbb{N}$. $\mathcal{O}^R_m$ is the set of possible $R_m$-outputs and $\mathcal{W}_m^R$ is the set of $R_m$-witnesses. We think of $b\in {\cal O}^R_m$ as being a solution to $a\in\{0,1\}^m$ for $R_m$ if $R_m(a,b;c)$ holds for all $c\in {\cal W}^R_m$. 
\end{notation}

A weaker notion of reducibility between $\TF\Sigma_2$ problems has been considered, known as \emph{counter-example reducibility} \cite{Thapen2024}.
Such a reduction from a search problem $Q$ to a search problem $R$ can be interpreted as the following game: Suppose Alice wants to solve $Q$ on input $x$. She heard that Bob is actually able to solve $R$ using some solver $S$, and thus, using the reduction, she sends $f(x)$ to Bob. In turns, Bob uses his solver $b:=S\big(f(x)\big)$ and Alice transforms $b$ into the output $y:= g(x,b)$. Now, if $y$ is indeed an output for $x$, Alice got what she wanted. But let us say that Alice does not trust the solution and, after some searching, is able to find $z$ such that $Q(x,y;z)$ does not hold (and hence $y$ is not a correct output for $x$). 
She confronts Bob about the issue, claiming that Bob's solver is not correct. 
Bob, who is very proud of his solver, dismisses her by asking for a proof.
Then by computing $c:= h(x,b,z)$, Alice is capable of finding a counter-example to $b$ being a correct output for $f(x)$, and hence proving that $S$ is incorrect. 

\begin{definition}\label{def:c_reduc}
    Let $Q$ and $R$ be two $\TF\Sigma_2^{dt}$ problems.
    A \emph{counter-example reduction from $Q$ to $R$} or a \emph{counter-example $R$-formulation of $Q$} is an $R$-formulation $(f,g)$ along with, for each $b\in \mathcal{O}^R_{s(n)}$ and $z\in\mathcal{W}^Q_n$, a function $h_{b,z}:\{0,1\}^n\rightarrow \mathcal{W}^R_{s(n)}$ such that $$R_{s(n)}\big(f(x),b;h_{b,z}(x)\big)\Longrightarrow Q_n\big(x,g_b(x);z\big)$$
    where $s(n)$ is the size of the reduction.
    The depth, and complexity of the formulation are defined as in \autoref{def:many_one_red}.
\end{definition}
By abuse of notation, an $R$-formulation of a problem $Q$ will also be said to be a counter-example reduction if there exists a function $h$ such that the triplet $(f,g,h)$ is a counter-example reduction.
We also define another type of reduction that is the translation of the $\Sigma_2$-weakenings.
\begin{definition}\label{def:weakening4}
    Let $Q$ and $R$ be two $\TF\Sigma_2^{dt}$ search problems and let $(f,g)$ be an $R$-formulation of $Q$.
    We say that $(f,g)$ is a \emph{weakening} if 
    \begin{itemize}
        \item $f$ is the identity,
        \item for each $b$, $g_b$ is a constant function.
    \end{itemize}
    In particular, weakenings are formulations that are always of size $n$ and depth $1$.
\end{definition}

\begin{proposition}\label{prop:weaknening_equiv}
    Let $Q$ and $R$ be two $\TF\Sigma_2^{dt}$ search problems.
    Then the following are equivalent:
    \begin{itemize}
        \item There exists a weakening from $Q$ to $R$,
        \item $F_{R_n}$ is a $\Sigma_2$-weakening of $F_{Q_n}$ for every $n$.
    \end{itemize}
\end{proposition}
For a decision tree $T(x)$ whose leaves are labeled with $0$ and $1$, we will abuse notation and also denote by $T(x)$ the DNF obtained from taking the OR over the terms corresponding to its accepting paths.
We also write $\overline{T}(x)$ for the function $1-T(x)$. Its decision tree is the same as $T$ but with inverted labels and its DNF is obtained from taking the OR over the terms corresponding to the rejecting paths of $T$.
\begin{proof}[Proof of \autoref{prop:weaknening_equiv}]
    Writing $s := s(n)$, say we have a weakening $R$-formulation $(f,g)$ of $Q$.
    Then for each $b\in\mathcal{O}^R_{s}$ the axiom $\bigvee_c\overline{R}_{n}(x,b;c)$ is a $\Sigma_2$-weakening of $\bigvee_z\overline{Q}_{n}(x,g_b;z)$ by the correctness of the reduction, where $c$ ranges over $\mathcal{W}^R_{s}$ and $z$ over $\mathcal{W}^Q_n$.
    If $F_{R_n}$ is a $\Sigma_2$-weakening of $F_{Q_n}$, then set $g_b=y_b$ such that $\bigvee_c\overline{R}_{n}(x,b;c)$ is a weakening of $\bigvee_z\overline{Q}_{n}(x,y_b,z)$.
\end{proof}
Considering only weakenings and counter-example reductions seems very restrictive, as the correctness of counter-example reductions is efficiently verifiable (since they actually occur at the $\TFNP$ level as highlighted in \cite{Thapen2024}) and weakenings are somewhat trivial.
Surprisingly enough, they seem to capture reductions in their entirety in the sense of the following lemma.
\begin{lemma}
\phantomsection\label{lem:factorization}
    Let $Q$ and $R$ be two $\TF\Sigma_2^{dt}$ problems and let $(f,g)$ be an $R$-formulation of $Q$.
    Then there exists $P\in\TF\Sigma_2^{dt}$ such that $(f,g)$ decomposes into two reductions $Q\rightarrow P\rightarrow R$ and:
    \begin{itemize}
        \item the reduction $Q\rightarrow P$ is a weakening,
        \item the reduction $P\rightarrow R$ is a counter-example reduction.
    \end{itemize}
\end{lemma}
This decomposition lemma  mirrors the characterization given by \hyperref[thm:wRACharact]{Theorem~\ref{thm:wRACharact}}: One part corresponds to taking a weakening of a formula while the other has a validity that can be efficiently verified (on one hand $\SA$-refutations and on the other counter-example reductions---See more about this in \hyperref[subsect:c_reduc]{Section \ref{subsect:c_reduc}}). 
In order to prove this lemma, we first need to introduce the \emph{reduced problem} associated with a reduction, as it will be our choice for $P$.

For a search problem $R_m(a,b;c)$, we write $R_{m,b,c}(a)$ for the decision tree computing $R_m(a,b;c)$ to put an emphasis on the fact that its only variable is $a$.
\begin{definition}\label{def:reduced_pb}
    Let $Q$ and $R$ be two $\TF\Sigma_2^{dt}$ search problems. 
    For $(f,g)$ an $R$-formulation of $Q$ of size $s:=s(n)$, the \emph{reduced problem $R(f,g)$} is the sequence of ternary relations $R(f,g)_n\subseteq\{0,1\}^n\times (\mathcal{O}^Q_n\times \mathcal{O}^R_{s})\times \mathcal{W}^R_{s}$ given by $$R(f,g)_n\big(x,(y,b);c\big) \Longleftrightarrow g_b(x) = y\ \mathrm{and}\ R_{s}\big(f(x),b;c\big).$$
    The formula associated to the reduced problem $F_{R(f,g)_n}$ has axioms $$\overline{G}_{b,y}(x)\vee\bigvee_c\overline{R}_{s,b,c}\circ f(x)$$ and $G_{b,y}(x)$ is the indicator function for the event $g_b(x) = y$. We see $G_{b,y}(x)$ as the decision tree (and hence can also be thought of as a DNF) obtained from the decision tree computing $g_b$ and replacing the labels of its leaf with 1 if the leaf is labeled with $y$ and 0 otherwise. 
\end{definition}
The formula associated to the reduced problem $R(f,g)$ states ``the reduction given by $(f,g)$ is not correct'', and for the tuple $(b,y)$, the meaning of the corresponding axiom is ``either $g_b(x)\neq y$, or there is a counter-example $c$ to $b$ being an $R$-output for $f(x)$''. The formula encodes the following procedure: Say we are given a $Q$-input $x$ and a purported output $y$. To verify that $y$ is a $Q$-output for $x$, instead of computing $Q(x,y;z)$ for each $z$, one can do the following: compute $f(x)$ and, given $b$ an $R$-output for $f(x)$, compute $R(f(x),b;c)$ for each $c$ and verify that $g_b(x) = y$. 

\begin{proof}[Proof of {\hyperref[lem:factorization]{Lemma~\ref*{lem:factorization}}}]
    Let $Q$, $R$, and $(f,g)$ be as in the statement of the lemma, set $P$ to be $R(f,g)$ and denote by $s := s(n)$ the size of the reduction.
    The formula $F_{R(f,g)_n}$ has an axiom $\overline{G}_{b,y}(x)\vee\bigvee_{c}\overline{R}_{s,b,c}\circ f(x)$ for every pair $y,b$, where the variable $c$ ranges over $\mathcal{W}^R_{s}$, and $F_{Q_n}$ has an axiom $\bigvee_{z}\overline{Q}_{n,y,z}(x)$ for each $y$ where the variable $z$ ranges over $\mathcal{W}^Q_n$.   
    By the correctness of the reduction (if $y=g_b(x)$ and $b$ is an $R$-output for $f(x)$, then $y$ is a $Q$-output for $x$), we get that the formula $$\bigvee_{z}\overline{Q}_{n,y,z}(x)\Longrightarrow \overline{G}_{b,y}(x)\vee\bigvee_{c}\overline{R}_{s,b,c}\circ f(x)$$ is a tautology and $F_{R(f,g)_n}$ is indeed a $\Sigma_2$-weakening of $F_{Q_n}$. 
    
    Now, setting $f'(x) = f(x)$ for all $x$, $g'_b(x) = \big(g_b(x),b\big)$ for all $b$, and $h'_{b,c}(x) = c$ for each pair $b,c$, the triplet $(f',g',h')$ forms a counter-example $R$-formulation of $R(f,g)$ since for every $x,y,b$, and $c$ we get $$R_{s}\big(f'(x), b;h'_{b,c}(x)\big)=R_{s,b,c}\circ f(x)=R(f,g)_n\big(x,(b,g_b(x));c)=R(f,g)_n\big(x,g'_b(x);c\big). \qedhere$$
\end{proof}

\subsection{Proving \autoref{lem:LOP_non_reduc}}

\begin{proposition}\label{prop:SA_proof_reduced_problem}
    Let $Q$ and $R$ be two $\TF\Sigma_2^{dt}$ search problems and let $(f,g)$ be an $R$-formulation of $Q$ of size $s(n)$ and depth $d(n)$.
    If $F_{R_m}$ admits an $\SA$-refutation of degree $d'(m)$, then $F_{R(f,g)_n}$ admits one of degree $d(n)(d'(s(n))+2)$.
\end{proposition}
In the following proof, for a polynomial $p(a)= p(a_1,\ldots,a_m)$ over $m$ variables and a collection of $m$ decision trees $f(x)=(f_i(x))_{i\in[m]} $, the polynomial $p\circ f(x)$ is $p\big(f_1(x),\ldots,f_m(x)\big)$ where each $f_i$ is seen as the sum of the monomials representing its accepting paths. Note that when $p$ is a junta, the resulting polynomial can also be represented by a junta. Indeed, for a term $t(a)$, one can compute $t\circ f(x)$ by querying $f_i$ if the variable $a_i$ appears in $t$. This gives a decision tree that can in turn be transformed into a conical junta. The degree of $p\circ f(x)$ is equal to $\deg(p)\cdot\depth(f)$.
We also define $T\circ f(x)$ a similar way when $T$ is a decision tree and its depth is equal to $\depth(T)\cdot\depth(f)$.
\begin{proof}
    Writing $s:= s(n)$, let us first recall the axioms (in their polynomial inequality form) of the two formulas.
    \begin{itemize}
        \item Axioms of $F_{R_s}$: $\sum_c\overline{R}_{s,b,c}(a)-1\geq 0$ for each $b\in \mathcal{O}^R_{s}$ with $c$ ranging over $\mathcal{W}^R_{s}$;
        \item Axioms of $F_{R(f,g)_n}$: $\overline{G}_{b,y}(x) + \sum_c\overline{R}_{s,b,c}\circ f(x)-1\geq 0$ for each pair $b\in\mathcal{O}^R_{s}$ and $y\in\mathcal{O}^Q_n$ with $c$ ranging over $\mathcal{W}_{s}^R$. Recall that each axiom encodes the fact that either $g_b(x)\neq y$ or $c$ witnesses that $b$ is not an $R$-output of $f(x)$ and their conjunction contradicts the correctness of the reduction.
    \end{itemize}
    As the axioms of $F_{R(f,g)_n}$ closely resemble the ones of $F_{R_{s}}$, our goal is to transform the $\SA$-refutation of the later formula to one of the former.
    The idea is pretty simple: compose every polynomial appearing in the refutation of $R_{s}$ with the function $f$ and the function computed by the sum should still be $-1$.
    Let $\sum_bJ_b(a)(\sum_c\overline{R}_{s,b,c}(a) -1) + J(a) = -1$ be an $\SA$-refutation of $R_{s}$.
    We explain how to deduce an $\SA$-refutation step by step. 
    \begin{enumerate}
        \item For each axiom $\overline{G}_{b,y}(x) + \sum_c\overline{R}_{s,b,c}\circ f(x) -1$, multiply by $G_{b,y}(x)$. Since $T(x)\cdot\overline{T}(x) = 0$ for any decision tree, we get $G_{b,y}(x)\cdot(\sum_c\overline{R}_{s,b,c}\circ f(x)-1)$ for each tuple $(b,y)$.
        \item Using the fact that $\sum_y G_{b,y} = 1$ (as we sum over the indicator variables of all paths of the decision tree $g_b$, and every assignment must be consistent with at least one of those paths), summing over $y$ and multiplying by $J_b\circ f$ gives us $$J_b\circ f(x)\Big(\sum_c\overline{R}_{s,b,c}\circ f(x)-1\Big).$$
        \item Finally, summing over all $b$ and summing $J\circ f(x)$, the result becomes $$\sum_bJ_b\circ f(x)\Big(\sum_c\overline{R}_{s,b,c}\circ f(x)-1\Big)+J\circ f(x)$$
    \end{enumerate}
    This actually gives a valid $\SA$-refutation since we (1) Multiply our axioms by conical juntas, (2) Sum all the results, and (3) Add an extra junta to the sum.
    Looking at this polynomial as a function over the boolean cube, we get the following: let $\chi$ be an assignment of $x$ and $\alpha$ be the resulting assignment $f(\chi)$ of $a$.
     Then $$\sum_bJ_b\circ f(\chi)\Big(\sum_c\overline{R}_{s,b,c}\circ f(\chi)-1\Big)+J\circ f(\chi) = \sum_bJ_b(\alpha)\Big(\sum_c\overline{R}_{s,b,c}(\alpha)-1\Big)+J(\alpha) =-1.$$
    Since functional equality is the same as polynomial equality in multi-linear arithmetic we get the equality $\sum_bJ_b\circ f(x)\Big(\sum_c\overline{R}_{s,b,c}\circ f(x)-1\Big)+J\circ f(x)=-1.$
    Hence, setting $J'_{b,y}(x) = J_b\circ f(x)\cdot G_{b,y}(x)$ and $J'(x) = J\circ f(x)$, we get the $\SA$-refutation $$\sum_{b,y}J'_{b,y}(x)\Big(\overline{G}_{b,y}(x)+\sum_c\overline{R}_{s,b,c}\circ f(x) -1\Big)+ J'(x)=-1$$ of $F_{R(f,g)_n}.$
    The degree of the refutation is $d(n)(d'(s(n))+2)$.
\end{proof}
Using this last result, we are able to prove \autoref{lem:LOP_non_reduc}.
\begin{proof}[Proof of \autoref{lem:LOP_non_reduc}]
    Let $R$ be a $\TF\Sigma_2^{dt}$ search problem such that $R_m$ admits a $\SA$-refutation of poly-logarithmic degree $d'(m)$ and let $(f,g)$ be an $R$-formulation of $\LOP$ of size $s(n)$ and degree $d(n)$.
    Then by \hyperref[lem:factorization]{Lemma~\ref*{lem:factorization}}, we have an $R$-formulation of similar size and depth of $R(f,g)$ and by the same lemma as well as \autoref{prop:SA_proof_reduced_problem} and the proof of \autoref{thm:LOPnotsRA}, we have
    \begin{itemize}
        \item $R(f,g)_n$ is a weakening of $\LOP_n$ for each $n$ and hence admits a super-poly-logarithmic lower bound on the degree of its $\SA$-refutations;
        \item $R(f,g)_n$ admits an $\SA$-refutation of degree $d(n)(d'(s(n))+2)$.
    \end{itemize}
    We then conclude that either $s(n)$ is super quasi-polynomial or either $d(n)$ is super poly-logarithmic since the sum $d(n)(d'(s(n))+2)$ is super poly-logarithmic and $d'$ is poly-logarithmic and hence the $R$-formulation is not efficient.
\end{proof}

\section{Open Problems}
\paragraph{Separating LeastNumber from StrongAvoid.}
Can we show that $\Min \not \leq_{dt} \sRA$?
While $\Min$ has a low-degree $\SA$ proof, the size of this proof is exponential, and therefore this does not imply a reduction to $\sRA$. The existence of a low-degree $\SA$ proof precludes us from using the techniques in this paper to obtain this separation. Furthermore, to our knowledge, all of the size lower bounds that have been proven for Sherali-Adams proceed by first proving a degree lower bound via pseudo-expectations, and then applying the size-degree tradeoff \cite{PitassiS12}. The only lower bounds which do not proceed via this strategy apply strictly to $\SA$ \cite{GoosHJMPRT22,RezendePR23}, they do not apply to the general Sherali-Adams proof system. 
Can these techniques be used in order to separate $\Min$ from $\sRA$?

\paragraph{Counter-Example Reductions and Herbrandization.}\phantomsection\label{subsect:c_reduc}

Thapen \cite{Thapen2024} gave a notion of \emph{herbrandization} for $\TF\Sigma_2^{dt}$ problems, which associates a $\TF\Sigma_2$ problem $R$ with a $\TFNP$ problem $\mathrm{Checkable}(R)$ satisfying the following property.
\begin{lemma}[\cite{Thapen2024}]\label{lem:c_reduc_and_tfnp}
    Let $Q, R \in \TF\Sigma_2^{dt}$.
    If $Q$ admits low-complexity counter-example $R$-formulation, then $\mathrm{Checkable}(Q)$ admits a low-complexity $\mathrm{Checkable}(R)$-formulation.
\end{lemma}
This previous lemma, coupled with our decomposition (\hyperref[lem:factorization]{Lemma~\ref*{lem:factorization}}), implies the following. 
    \begin{lemma}
        If for every weakening $P$ of $Q \in \TF\Sigma_2^{dt}$, $\mathrm{Checkable}(P)\nleq_{dt} \mathrm{Checkable}(R)$, then $Q\nleq_{dt}R$.
    \end{lemma}
In light of this, if $P$ is a weakening of $Q$, then it is natural to ask whether there is anything that we can infer about the relationship between $\mathrm{Checkable}(Q)$ and $\mathrm{Checkable}(P)$?
Heuristically, to obtain intuition as to whether $Q$ reduces to $R$ it has been useful to look at the relationship between $\mathrm{Checkable}(Q)$ and $\mathrm{Checkable}(R)$ in $\TFNP^{dt}$.
 This question asks to what degree this heuristic can be formalized. 

Decision-tree reductions between problems in $\TFNP$ and proofs in certain polynomial-time verifiable proof systems are tightly connected. Counter-example reductions are a polynomial-time verifiable way to relate search problems in higher levels of the polynomial hierarchy. We ask whether one can also obtain characterizations of counter-example reductions to $\TF\Sigma_2^{dt}$ problems by proof systems. Furthermore, if this is the case, how do these proof systems compare to the ones obtained by first herbrandizing the $\TF\Sigma_2$ problem $R$ to a $\TFNP$ problem $\mathrm{Checkable}(R)$ and then taking its corresponding proof system via \cite{BussFI23}.

\section{Acknowledgements}
We thank the anonymous referees for helpful comments.\\
Noah Fleming was supported by the Swedish Research Council under grant number 2025-06762.

\bibliographystyle{alpha}
\bibliography{biblio}
\end{document}